\documentclass[11pt]{amsart}

\usepackage{amsmath,amssymb,amsthm}
\usepackage{mathtools}
\usepackage{mathrsfs}

\usepackage[T1]{fontenc}
\usepackage{lmodern}

\usepackage{xcolor}
\usepackage[
    colorlinks=true,
    citecolor=red,
    linkcolor=blue,
    urlcolor=blue
]{hyperref}

\newtheorem{theorem}{Theorem}[section]
\newtheorem{lemma}[theorem]{Lemma}
\newtheorem{proposition}[theorem]{Proposition}
\newtheorem{corollary}[theorem]{Corollary}

\theoremstyle{definition}
\newtheorem{definition}[theorem]{Definition}

\theoremstyle{remark}

\makeatletter
\renewcommand{\paragraph}{\@startsection{paragraph}{4}%
  \z@\z@{-\fontdimen2\font}%
  {\normalfont\bfseries}}
\makeatother

\title[Asymptotic Neutrality in Two--Center TFDW Theory]{Asymptotic Neutrality for the Ionization Conjecture in the Two--Center Thomas--Fermi--Dirac--von Weizs\"acker Theory}

\author{Lu Chen}
\address[Lu Chen]{Key Laboratory of Algebraic Lie Theory and Analysis of Ministry of Education, School of Mathematics and Statistics, Beijing Institute of Technology, Beijing, 100081, P.R. China}
\email{chenlu5818804@163.com}

\author{Yong Wu$^*$}
\address[Yong Wu]{Key Laboratory of Algebraic Lie Theory and Analysis of Ministry of Education, School of Mathematics and Statistics, Beijing Institute
of Technology, Beijing 100081, PR China}
\email{wu\_yong2025@163.com}



\thanks{$*$ Corresponding author.}

\thanks{The first author was partly supported by the National Natural Science
Foundation of China (No.~12622108) and a grant from Beijing Institute of
Technology (No.~2022CX01002).}

\date{}

\begin{document}

\begin{abstract}
The ionization conjecture for the Thomas--Fermi--Dirac--von Weizs\"acker model
asks how many electrons a given nuclear configuration can bind. So far, only
the single-center case has been settled by Frank, Nam, and Van Den
Bosch in~\cite{FrankNamVanDenBosch2018}, while the original molecular case has
still remained open. We advance this conjecture in the two-center case. We
prove that for two distinct nuclei of total charge
\(Z\), the number \(N\) of bound electrons satisfies
\(N\le Z+C Z^{153/154}(1+\log Z)^{3/8}\), with constants independent of the
nuclear positions and the nuclear charges. The proof relies on a
comparison with the neutral Thomas--Fermi density and on localization and
exterior \(L^1\)-estimates coupled through the screened potential. Our results
could be extended to molecules with any fixed number \(K\) of nuclei.
\end{abstract}

\maketitle

\tableofcontents

\section{Introduction}

A fundamental question in quantum many-body theory is how many electrons a
system
of nuclei can bind. For a nucleus of charge \(Z\), the ionization problem asks
for the
maximum negative ionization, that is, the largest number of electrons that can
be
bound when this number is allowed to exceed \(Z\).

\paragraph{Ionization conjecture.} It asserts that a nucleus of charge \(Z\)
can bind at most \(Z+C\) electrons and that a molecule of \(K\) nuclei can bind at
most \(Z+CK\), where the constant \(C\) is universal in both cases.

In nonrelativistic many-body Schr\"odinger theory, Lieb~\cite{Lieb1984} proved
that
the maximum number \(N_c\) of negative particles that can be bound to a nucleus
of
charge \(Z\) satisfies
\[
    N_c<2Z+1,
\]
and that for a molecule of \(K\) fixed nuclei with total nuclear charge \(Z\),
\[
    N_c<2Z+K .
\]
The bound depends neither on the number of bound particles nor on their
statistics. In particular, it applies to bosons as well as fermions. For
fermionic
atoms and every \(Z>0\), Nam~\cite{Nam2012} subsequently improved Lieb's upper
bound to
\[
    N_c<1.22Z+3Z^{1/3},
\]
which is smaller than \(2Z+1\) when \(Z\ge6\).

A complementary line of research has brought the coefficient of the leading term
\(Z\) in the ionization upper bound down to the neutral value \(1\). Writing
\(N_c(Z)\) for the maximum number of electrons that can be bound to a nucleus of
charge \(Z\), Lieb, Sigal, Simon, and Thirring~\cite{LiebSigalSimonThirring1988}
first proved the asymptotic neutrality of fermionic atoms,
\[
    \lim_{Z\to\infty} N_c(Z)/Z=1 .
\]
Fefferman and Seco~\cite{FeffermanSeco1990} quantified the rate of convergence
and
proved
\[
    N_c(Z)=Z+O(Z^{47/56}),
\]
and Seco, Sigal, and Solovej~\cite{SecoSigalSolovej1990} obtained the same bound
on
the excess charge \(Q=N_c(Z)-Z\) from a simple bound on the ionization energy.
A central ingredient in these proofs is a comparison with Thomas--Fermi theory
together with estimates that control the excess charge outside a ball of radius
\(R\)~\cite{FeffermanSeco1990,SecoSigalSolovej1990}. For a modern formulation of
such exterior estimates in TFDW theory we refer
to~\cite{FrankNamVanDenBosch2018}.

Thomas--Fermi (TF) theory provides the variational framework for this
comparison.
It is the oldest density functional theory: it replaces the many-body wave
function
by the electron density and captures the leading-order behavior of the many-body
ground state energy, so that the many-electron problem becomes a minimization
problem for a density functional.
Lieb and Simon~\cite{LiebSimon1977} established the rigorous mathematical
foundations
of TF theory for atoms and molecules, including the existence and uniqueness of
the
neutral minimizer and its characterization by the TF equation.
For this minimizer the total number of electrons equals the total nuclear charge
\(Z\),
so it is the natural reference density for systems with large nuclear charge.

Quantum corrections to TF theory come from the von Weizs\"acker gradient term
\[
    c_{\mathrm W}\int_{\mathbb R^3}|\nabla\sqrt{\rho_0}|^2 \,dx
\]
and the Dirac exchange term
\[
    -c_{\mathrm D}\int_{\mathbb R^3}\rho^{4/3}\,dx,
\]
leading to the TFDW functional.

In Thomas--Fermi--von Weizs\"acker (TFW) theory, only the gradient correction is
included. Negative ions exist in this theory and the excess charge admits a
uniform
bound: Benguria and Lieb~\cite{BenguriaLieb1985} proved that for a system of
\(K\)
nuclei with total nuclear charge \(Z\), the maximum number \(N_c\) of electrons
that
can be bound satisfies
\[
    N_c\leq Z+CK,
\]
where \(C\) is a universal constant and is independent of the nuclear charges
and
positions. This result holds for both atoms and molecules. It shows that for a
fixed
number of nuclei the excess charge \(N_c-Z\) remains bounded even as the total
nuclear charge
tends to infinity.

For the TFDW model the mechanism is different. The exchange term prevents a
detached
part of the electron density from spreading out, while the free part of the
energy is
translation invariant. Consequently excess electrons can escape to infinity at
no
energy cost. This makes the existence of a minimizer in the regime \(N\leq Z\)
and the
nonexistence of minimizers for large \(N\) two distinct phenomena, and the
latter has
to be excluded by an additional hypothesis.

Lions~\cite{Lions1987} proved, in particular, that the TFDW minimization problem
has a
minimizer whenever \(N\leq Z\), and this holds for a system of \(K\) nuclei of
total
nuclear charge \(Z\).
Lu and Otto~\cite{LuOtto2014} proved that, without an external potential, no
minimizer
exists when the electron number is sufficiently large.
Nam and Van Den Bosch~\cite{NamVanDenBosch2017} proved that, for a small total
nuclear
charge, no minimizer exists when the electron number is sufficiently large.
Frank, Nam, and Van Den Bosch~\cite{FrankNamVanDenBosch2018} proved that in
atomic TFDW
theory a nucleus of charge \(Z>0\) can bind at most \(Z+C\) electrons, where
\(C\) is a
universal constant.
The atomic case is therefore settled, and we may assume that a minimizer exists.
The
difficulty in the two-center problem is that the comparison density and the
exterior
estimate must be uniform in the nuclear positions and in the nuclear charges.

The ionization conjecture concerns molecular systems with several nuclei in
the first place. The atomic case settled by Frank, Nam, and Van Den Bosch
in~\cite{FrankNamVanDenBosch2018} is a special case of it, and the molecular
case has remained largely open. The
known bounds depend on the nuclear positions and charges, whereas the atomic
bound is uniform in \(Z\), and no asymptotic neutrality bound uniform in the
molecular geometry was available. We make progress on the two-center case
by proving such a bound.

The first difficulty is to obtain uniform estimates with respect to the
molecular
geometry. The nuclear positions may approach each other or become widely
separated, and
the charge ratio \(Z_1/Z_2\) may be arbitrary. Therefore arguments relying on a
fixed
separation scale or comparable nuclear charges cannot be applied directly, and
the
neutral TF density is not given in closed form, in contrast to the atomic case.
The
localization estimate below is designed to be uniform with respect to the
two-center
geometry.

The second difficulty concerns the comparison with the neutral Thomas--Fermi
density.
In TF theory the energy deficit controls the Coulomb energy of the difference
\(\rho_0-\rho^{\mathrm{TF}}\) directly, while in TFDW the exchange term
contributes a
positive local energy density, so that the deficit alone yields no such control
and the
comparison has to be organised as a bootstrap argument.

Another difficulty is the possible escape of electrons to infinity. When the
electron
number exceeds the total nuclear charge, excess electrons may escape from the
nuclear
region and contribute to the electron number in the exterior region. The
exterior estimate must
therefore control the number of electrons outside a ball through the screened
potential
\(\Phi_r=V-U_r\), which is itself the quantity the second difficulty is about,
so
that the two have to be treated simultaneously.

These difficulties motivate the development of new estimates that combine
Thomas--Fermi comparison with localization arguments adapted to the two-center
setting.

The main contribution of this paper is to establish a quantitative asymptotic
neutrality estimate for the two-center TFDW theory. The constants \(C\) and
\(Z_0\) in
Theorem~\ref{thm:main} depend only on the coefficients of the functional and are
independent of the distinct nuclear positions, of the nuclear charges, of \(N\),
and of
the minimizer. The proof rests on three ingredients: a quantitative comparison
with the
neutral Thomas--Fermi density, which bounds the Coulomb energy
\(\mathcal{D}(\rho_0-\rho^{\mathrm{TF}},\rho_0-\rho^{\mathrm{TF}})\) of the
difference between the
minimizer and the neutral TF minimizer; an interior estimate for the Coulomb
potential that is uniform with respect to the two-center geometry; and an
exterior
\(L^1\)-estimate for the number of electrons in the exterior region. The last
two are
coupled through the screened potential \(\Phi_r\), whose decay is a consequence
of the comparison with TF theory.

Although we only treat two nuclei, the method could
be extended to any fixed number \(K\) of nuclei and yield
\(N\le Z+C(K)Z^{153/154}(1+\log Z)^{3/8}\),
where \(C(K)\) depends on \(K\) as well.

The paper is organized as follows. In Section~\ref{sec:main-results}, we state
the main results and collect the notation. In Section~\ref{sec:rough-bound}, we
prove a rough bound on the number of electrons. In
Section~\ref{sec:tf-comparison}, we compare the minimizer with the neutral
Thomas--Fermi density. In Section~\ref{sec:coulomb-truncation}, we establish the
interior estimate for the Coulomb potential. In
Section~\ref{sec:exterior-estimates}, we prove the exterior \(L^1\)-estimates for
the electron number. Finally, in Section~\ref{sec:main-proof}, we complete the
proof of the main theorem.

\section{Main Results and Notation}
\label{sec:main-results}

In this section we state the main results and introduce the notation used in the
sequel.

\subsection{The two-center TFDW problem and the main result}

Throughout the paper, we fix coefficients \(c_{\mathrm{TF}}>0\), \(c_{\mathrm W}>0\), and
\(c_{\mathrm D}\ge0\). Unless a dependence is explicitly indicated, constants
\(C,c>0\) depend only on these coefficients and may change from line to line.
They do
not depend on the nuclear charges, \(N\), or the minimizer.

Let two nuclei be located at distinct points \(R_1,R_2\in\mathbb R^3\), with
nuclear
charges \(Z_1,Z_2>0\). We write
\[
    Z=Z_1+Z_2
\]
for the total nuclear charge. The corresponding nuclear potential is
\[
    V(x)=\frac{Z_1}{|x-R_1|}+\frac{Z_2}{|x-R_2|}.
\]

We write \(\mathcal{D}(\rho,\sigma)\) for the Coulomb interaction
\[
    \mathcal{D}(\rho,\sigma)=\frac{1}{2}\int_{\mathbb R^3\times \mathbb R^3}
    \frac{\rho(x)\sigma(y)}{|x-y|} dxdy,
\]
whenever the integral is absolutely convergent.

\begin{definition}(TFDW energy functional)
    For nonnegative densities \(\rho\) satisfying
    \[
        \sqrt{\rho}\in H^1(\mathbb R^3),
    \]
    we define the two-center TFDW energy by
    \[
        \begin{aligned}
            \mathcal{E}_V^{\mathrm{TFDW}}(\rho)
            &=
            c_{\mathrm{TF}}
            \int_{\mathbb{R}^3}\rho(x)^{5/3}\,dx
            -
            \int_{\mathbb{R}^3} V(x) \rho(x)\,dx
            +
            \mathcal{D}(\rho,\rho)
            \\
            &+
            c_{\mathrm W}
            \int_{\mathbb{R}^3}
            \left|\nabla\sqrt{\rho(x)}\,dx\right|^2\,dx
            -
            c_{\mathrm D}
            \int_{\mathbb{R}^3}
            \rho(x)^{4/3}\,dx.
        \end{aligned}
    \]
\end{definition}

All terms are finite on this class, by the constraint \(\int_{\mathbb R^3}\rho\,dx=N\).
Sobolev interpolation gives \(\rho\in L^p(\mathbb R^3)\) for \(1\le p\le3\), so
that
\(\rho\in L^{5/3}(\mathbb R^3)\cap L^{1}(\mathbb R^3)\). The
Hardy--Littlewood--Sobolev
inequality bounds \(\mathcal{D}\) on \(L^{6/5}(\mathbb R^3)\), while
Hardy's
inequality followed by Cauchy--Schwarz bounds the nuclear attraction.

For a prescribed electron number \(N\ge0\), the corresponding variational
problem is
    \[
        E_V^{\mathrm{TFDW}}(N)
        =
        \inf
        \left\{
        \mathcal{E}_V^{\mathrm{TFDW}}(\rho):
        \rho\geq0,\,
        \sqrt{\rho}\in H^1(\mathbb{R}^3),\,
        \int_{\mathbb{R}^3}\rho(x)\,dx=N
        \right\}.
    \]

A density is called a minimizer if it attains the infimum in
\(E_V^{\mathrm{TFDW}}(N)\). From now on we reserve the symbol \(\rho_0\) for a
minimizer.

We now state the main result of this paper concerning the quantitative
asymptotic neutrality of the two-center TFDW model.

\begin{theorem}\label{thm:main}
(Asymptotic neutrality)
There exist constants \(C>0\) and \(Z_0\ge2\), depending only on
\(c_{\mathrm{TF}},c_{\mathrm D},c_{\mathrm W}\), such that for every pair of
distinct points \(R_1,R_2\in\mathbb R^3\), every \(Z_1,Z_2>0\) with
\(Z=Z_1+Z_2\ge Z_0\), and every \(N\ge0\), if
\(E_V^{\mathrm{TFDW}}(N)\) has a minimizer, then
\[
N \leq Z + C Z^{153/154} (1+\log Z)^{3/8}.
\]
\end{theorem}

Since \(Z^{-1/154}(1+\log Z)^{3/8}\to0\) as \(Z\to\infty\), we obtain
\[
    N\leq Z+o(Z).
\]
The constants \(C\) and \(Z_0\) are uniform in the nuclear positions and in the
nuclear charges.

\subsection{Notation}

We collect here the notation used throughout the paper.

\paragraph{Two-center geometry.}
Let
\[
    \delta(x)=\min\{|x-R_1|,|x-R_2|\},
\]
so that \(\delta(x)\) is the distance from \(x\) to the nearer of the two
nuclei. We write
\[
    \Omega_1=\{x\in\mathbb R^3:|x-R_1|\leq|x-R_2|\},
    \qquad
    \Omega_2=\{x\in\mathbb R^3:|x-R_2|<|x-R_1|\}
\]
for the two regions consisting of the points that are closer to \(R_1\) and to
\(R_2\),
respectively. They are half-spaces separated by the perpendicular bisector plane
of the
segment \(R_1R_2\), and they partition \(\mathbb R^3\) up to a plane of measure
zero. For
\(r>0\) we call
\[
    K_r=\{x\in\mathbb R^3:\delta(x)\le r\}
\]
the interior region and
\[
    K_r^c=\{x\in\mathbb R^3:\delta(x)>r\}
\]
the exterior region, so that \(K_r\) is the union of the two closed balls of
radius \(r\)
centered at the nuclei. We also write
\[
    \ell=|R_1-R_2|>0
\]
for the distance between the nuclei.

\paragraph{Densities and potentials.}
Let \(\rho_0\) be a TFDW minimizer and let \(\rho^{\mathrm{TF}}\) be the neutral
Thomas--Fermi density of nuclear charge \(Z\). We write
\[
    U(x)=\int_{\mathbb R^3}\frac{\rho_0(y)}{|x-y|}\,dy,
    \qquad
    U^{\mathrm{TF}}(x)=\int_{\mathbb R^3}\frac{\rho^{\mathrm{TF}}(y)}{|x-y|}\,dy,
\]
and, for \(r>0\),
\[
    U_r (x)=\int_{K_r}\frac{\rho_0(y)}{|x-y|}\,dy,
    \qquad
    U_r^{\mathrm{TF}}(x)=\int_{K_r}\frac{\rho^{\mathrm{TF}}(y)}{|x-y|}\,dy .
\]
Away from the nuclei, we define the screened potentials
\[
    \Phi_r=V-U_r,
    \qquad
    \Phi_r^{\mathrm{TF}}=V-U_r^{\mathrm{TF}}.
\]

For \(r>0\), let
\[
\begin{aligned}
    e_r&=\mathbf 1_{\{\delta>r\}}\rho_0,
    &\qquad
    c_r&=\mathbf 1_{K_r}\rho_0,
    \\
    e_r^{\mathrm{TF}}&=\mathbf 1_{\{\delta>r\}}\rho^{\mathrm{TF}},
    &\qquad
    c_r^{\mathrm{TF}}&=\mathbf 1_{K_r}\rho^{\mathrm{TF}},
\end{aligned}
\]
so that \(\rho_0=e_r+c_r\) and
\(\rho^{\mathrm{TF}}=e_r^{\mathrm{TF}}+c_r^{\mathrm{TF}}\).
We also define
\[
    m_r=\int_{\{\delta>r\}}\rho_0\,dx,
    \qquad
    \mathcal{D}_r=\mathcal{D}(e_r,e_r),
\]
and
\[
    S_r=\sup_{\delta(x)\geq r}\delta(x)[\Phi_r(x)]_+ .
\]

\section{A Rough Bound on the Number of Electrons}
\label{sec:rough-bound}

In this section we prove a bound on \(N\) that does not use the comparison with
Thomas--Fermi theory. We begin with a one-center attraction estimate for a
density in
terms of its Thomas--Fermi and Coulomb energies, then we deduce an a priori
bound on
these two quantities for any minimizer. Finally we obtain the estimate
\[
    N\le 4Z+CZ^{2/3}+C .
\]

\subsection{A priori bounds for TFDW minimizers}

Since the minimizing energy is nonpositive in~\cite{NamVanDenBosch2017}, every
minimizer satisfies
\[
    \mathcal{E}_V^{\mathrm{TFDW}}(\rho_0)
    =
    E_V^{\mathrm{TFDW}}(N)
    \leq0 .
\]
Hence
\[
    \mathcal{D}(\rho_0,\rho_0)
    +
    \int_{\mathbb R^3}\rho_0^{5/3}
    \,dx\leq
    C\left(Z^{7/3}+N\right).
\]
This bound is independent of the distance between the two nuclei and of the
charge ratio,
and we use it in the sequel to control the Coulomb energy of the minimizer.

\begin{lemma}\label{lem:one-center-attraction}
Let \(\sigma\geq0\) satisfy
\(\sigma\in L^1(\mathbb R^3)\cap L^{5/3}(\mathbb R^3)\)
and \(\mathcal{D}(\sigma,\sigma)<\infty\). Then, for every \(R\in\mathbb R^3\),
\[
    \int_{\mathbb R^3}\frac{\sigma(x)}{|x-R|}\,dx
    \leq
    C
    \left(
    \int_{\mathbb R^3}\sigma^{5/3}
    \,dx\right)^{3/7}
    \mathcal{D}(\sigma,\sigma)^{1/7}.
\]
\end{lemma}

\begin{proof}
Let \(\mathcal{D}_\sigma=\mathcal{D}(\sigma,\sigma)\).
Without loss of generality, we assume \(\sigma\neq0\) and \(\mathcal{D}_\sigma>0\).
For every scalar \(s>0\), we define
\[
\psi_s(x)=\min\{|x-R|^{-1},s^{-1}\}.
\]
From H\"older's inequality, we obtain
\[
\begin{aligned}
\int_{\mathbb R^3}\frac{\sigma(x)}{|x-R|}\,dx
&=\int_{\mathbb R^3}\sigma(x)(|x-R|^{-1}-s^{-1})_+\,dx+\int_{\mathbb R^3}\sigma\psi_s\,dx\\
&\le C\|\sigma\|_{5/3}s^{1/5}+\int_{\mathbb R^3}\sigma\psi_s\,dx.
\end{aligned}
\]
For every \(\psi\in C_c^\infty(\mathbb R^3)\),
Fourier inversion and Cauchy--Schwarz inequality give
\begin{equation}\label{eq:fourier-parseval}
\begin{aligned}
\left|\int_{\mathbb R^3}\sigma\psi\,dx\right|
&=\frac{1}{(2\pi)^3}\left|\int_{\mathbb R^3}\hat\sigma(\xi)\overline{\hat\psi(\xi)}\,d\xi\right|\\
&=\frac{1}{(2\pi)^3}\left|\int_{\mathbb R^3}\frac{\hat\sigma(\xi)}{|\xi|}\,|\xi|\overline{\hat\psi(\xi)}\,d\xi\right|\\
&\leq\frac{1}{(2\pi)^3}\left(\int_{\mathbb R^3}\frac{|\hat\sigma(\xi)|^2}{|\xi|^2}\,d\xi\right)^{1/2}
\left(\int_{\mathbb R^3}|\xi|^2|\hat\psi(\xi)|^2\,d\xi\right)^{1/2}\\
&=\frac{1}{(2\pi)^3}\left(4\pi^2\mathcal{D}_\sigma\right)^{1/2}
\left((2\pi)^3\|\nabla\psi\|_2^2\right)^{1/2}\\
&=\left(\frac{\mathcal{D}_\sigma}{2\pi}\right)^{1/2}\|\nabla\psi\|_2\\
\end{aligned}
\end{equation}
Although \(\psi_s\notin L^2\), it belongs to the homogeneous Sobolev space
\(\dot H^1\), and
\[
\int_{\mathbb R^3}|\nabla\psi_s|^2\,dx=\frac{4\pi}{s}.
\]
Multiplying \(\psi_s\) by a sequence of smooth cutoffs \(\eta_k\), with
\(\eta_k=1\) in \(B_{R_k}\) and \(R_k\to\infty\), and then mollifying at a
scale tending to zero, yields \(\varphi_k\in C_c^\infty(\mathbb R^3)\) with
\(\varphi_k\to\psi_s\) and \(\nabla\varphi_k\to\nabla\psi_s\) in \(L^2\).
Applying~\eqref{eq:fourier-parseval} to \(\varphi_k\) and letting \(k\to\infty\)
extends it to \(\psi_s\), the pairings converging by dominated convergence
since \(0\le\varphi_k\le s^{-1}\) and \(\sigma\in L^1\).
Consequently,
\[
\int_{\mathbb R^3}\frac{\sigma(x)}{|x-R|}\,dx
 \le C\|\sigma\|_{5/3}s^{1/5}+\sqrt{2\mathcal{D}_\sigma/s}.
\]
Since the preceding bound holds for every \(s>0\), we choose \(s\) so that the
two terms balance and obtain
\[
\int_{\mathbb R^3}\frac{\sigma(x)}{|x-R|}\,dx
\le C\|\sigma\|_{5/3}^{5/7}\mathcal{D}_\sigma^{1/7}
=C\left(\int_{\mathbb R^3}\sigma^{5/3}\,dx\right)^{3/7}
\mathcal{D}(\sigma,\sigma)^{1/7},
\]
which is the asserted bound.
\end{proof}

\begin{proposition}\label{prop:hartree-apriori}
Let \(\rho_0\) be a minimizer. Then
\begin{equation}
\mathcal{D}(\rho_0,\rho_0)
+
\int_{\mathbb R^3}\rho_0^{5/3}
\,dx\leq
C\left(Z^{7/3}+N\right).
\end{equation}
\end{proposition}

\begin{proof}
Let
\[
    \rho_i=\rho_0\,\mathbf 1_{\Omega_i}
    \qquad i=1,2.
\]
Then
\[
\rho_0=\rho_1+\rho_2
\qquad\text{a.e.}
\]

For \(x\in\Omega_i\),
\[
    |x-R_i|
    \leq
    |x-R_j|,
    \qquad j\neq i,
\]
and hence
\[
V(x)
=
\frac{Z_1}{|x-R_1|}
+
\frac{Z_2}{|x-R_2|}
\leq
\frac{Z_1+Z_2}{|x-R_i|}
=
\frac{Z}{|x-R_i|}.
\]
Therefore
\[
\begin{aligned}
\int_{\mathbb R^3}V\rho_0
\,dx&=
\sum_{i=1}^2
\int_{\Omega_i}V\rho_0
\,dx\\
&\leq
Z
\sum_{i=1}^2
\int_{\mathbb R^3}
\frac{\rho_i(x)}{|x-R_i|}\,dx.
\end{aligned}
\]
From Lemma~\ref{lem:one-center-attraction}, we obtain
\[
\int_{\mathbb R^3}V\rho_0
\,dx\leq
CZ
\sum_{i=1}^2
A_i^{3/7} \mathcal{D}_i^{1/7},
\]
where
\[
A_i
=
\int_{\mathbb R^3}\rho_i^{5/3}\,dx,
\qquad
\mathcal{D}_i
=
\mathcal{D}(\rho_i,\rho_i).
\]

Since \(A_i,\mathcal{D}_i\geq0\),
\[
A_i^{3/7}\mathcal{D}_i^{1/7}
\leq
(A_i+\mathcal{D}_i)^{4/7}.
\]
Since the function \(t\mapsto t^{4/7}\) is concave,
\[
\sum_{i=1}^2
(A_i+\mathcal{D}_i)^{4/7}
\leq
C
\left(
\sum_{i=1}^2(A_i+\mathcal{D}_i)
\right)^{4/7}.
\]
Moreover,
\[
A_1+A_2
=
\int_{\mathbb R^3}\rho_0^{5/3}
\,dx\]
and
\[
    \mathcal{D}(\rho_0,\rho_0)
    =
    \mathcal{D}(\rho_1,\rho_1)
    +
    \mathcal{D}(\rho_2,\rho_2)
    +
    2\mathcal{D}(\rho_1,\rho_2)
    \geq
    \mathcal{D}(\rho_1,\rho_1)
    +
    \mathcal{D}(\rho_2,\rho_2).
\]

Thus
\[
\int_{\mathbb R^3}V\rho_0
\,dx\leq
CZ
\left(
\int_{\mathbb R^3}\rho_0^{5/3}
\,dx+
\mathcal{D}(\rho_0,\rho_0)
\right)^{4/7}.
\]

Let
\[
X
=
\int_{\mathbb R^3}\rho_0^{5/3}
\,dx+
\mathcal{D}(\rho_0,\rho_0).
\]
By Young's inequality,
for every \(\varepsilon>0\),
\[
Z X^{4/7}
\leq
\varepsilon X
+
C_\varepsilon Z^{7/3}.
\]
Hence, after choosing \(\varepsilon>0\) small enough,
\begin{equation}\label{eq:nuclear-absorbed}
\int_{\mathbb R^3}V\rho_0
\,dx\leq
\frac{c_{\mathrm{TF}}}{4}
\int_{\mathbb R^3}\rho_0^{5/3}
\,dx+
\frac14\mathcal{D}(\rho_0,\rho_0)
+
CZ^{7/3}.
\end{equation}

We next estimate the Dirac term. Young's
inequality gives, for every \(t\geq0\),
\begin{equation}\label{eq:dirac-young}
c_{\mathrm D}t^{4/3}
\leq
\frac{c_{\mathrm{TF}}}{4}t^{5/3}
+
\frac{c_{\mathrm D}^2}{c_{\mathrm{TF}}}t.
\end{equation}
Integrating~\eqref{eq:dirac-young}, we have
\[
c_{\mathrm D}
\int_{\mathbb R^3}\rho_0^{4/3}
\,dx\leq
\frac{c_{\mathrm{TF}}}{4}
\int_{\mathbb R^3}\rho_0^{5/3}
\,dx+
CN.
\]

Combining with~\eqref{eq:nuclear-absorbed}, we obtain
\[
\begin{aligned}
\mathcal{E}_V^{\mathrm{TFDW}}(\rho_0)
&\geq
\frac{c_{\mathrm{TF}}}{2}
\int_{\mathbb R^3}\rho_0^{5/3}
\,dx+
\frac34\mathcal{D}(\rho_0,\rho_0)\\
&\qquad
+c_{\mathrm W}
\int_{\mathbb R^3}|\nabla\sqrt\rho_0|^2
\,dx-
C(Z^{7/3}+N).
\end{aligned}
\]
Then
\[
\mathcal{E}_V^{\mathrm{TFDW}}(\rho_0)
\geq
c
\left(
\int_{\mathbb R^3}\rho_0^{5/3}
\,dx+
\mathcal{D}(\rho_0,\rho_0)
\right)
-
C(Z^{7/3}+N)
\]
for some \(c>0\).

Since \(\rho_0\) is a minimizer,
\[
    \mathcal{E}_V^{\mathrm{TFDW}}(\rho_0)
    =
    E_V^{\mathrm{TFDW}}(N)
    \leq0.
\]
Hence
\[
    \mathcal{D}(\rho_0,\rho_0)
    +
    \int_{\mathbb R^3}\rho_0^{5/3}
    \,dx\leq
    C(Z^{7/3}+N),
\]
which ends the proof.
\end{proof}

\subsection{A rough bound on \texorpdfstring{\(N\)}{N}}

Writing \(\ell=|R_1-R_2|>0\), without loss of generality, we may assume that
\[
    R_1=-\frac{\ell}{2}e,
    \qquad
    R_2=\frac{\ell}{2}e,
\]
where
\[
    e=(1,0,0).
\]

We decompose \(x\) as
\[
    x=te+x_{\perp}
\]
and define
\[
    F(x)=
    \left(
    \frac{\ell}{2}-|t|
    \right)e+x_\perp.
\]

Then \(F\) is \(1\)-Lipschitz, and
\[
    |F(x)|
    =
    \min\{|x-R_1|,|x-R_2|\}
    =
    \delta(x).
\]

Moreover, the restriction of \(F\) to \(\Omega_i\)
is an isometry.

Write \(E_0^{\mathrm{TFDW}}\) for the variational energy with zero external
potential. Nam and Van Den Bosch~\cite[Lemma~5]{NamVanDenBosch2017}
proved that, for every \(0\le M\le N\),
\begin{equation}\label{eq:binding-inequality}
\begin{aligned}
E_V^{\mathrm{TFDW}}(N)
&\le E_V^{\mathrm{TFDW}}(M)+E_0^{\mathrm{TFDW}}(N-M),
\\
E_0^{\mathrm{TFDW}}(m)
&<0 \qquad \text{when } m>0.
\end{aligned}
\end{equation}

Let \(f_0,f_1\) be nonnegative Lipschitz functions with
\(f_0^2+f_1^2=1\), and let \(\rho_0\) be a minimizer with \(\int_{\mathbb R^3}\rho_0\,dx=N\).
The densities \(f_j^2\rho_0\) are admissible.
Applying~\eqref{eq:binding-inequality} to their charges and using the IMS
identity gives
\begin{equation}\label{eq:partition-localization}
\begin{aligned}
\iint_{\mathbb R^3\times\mathbb R^3}\frac{f_0(x)^2\rho_0(x)f_1(y)^2\rho_0(y)}{|x-y|}\,dx\,dy
&\le \int_{\mathbb R^3} V f_1^2\rho_0
 \,dx+\alpha\int_{\{0<f_1<1\}}\rho_0\,dx\\
&\quad+c_{\mathrm W}\int_{\mathbb R^3}
 (|\nabla f_0|^2+|\nabla f_1|^2)\rho_0\,dx.
\end{aligned}
\end{equation}
Here \(\alpha=c_{\mathrm D}^2/(4c_{\mathrm{TF}})\).
The local TF--Dirac difference is bounded pointwise by
\(\alpha\rho_0\,\mathbf1_{\{0<f_1<1\}}\).
Indeed, since \(f_0^2+f_1^2=1\) and \(0\le f_j\le1\),
\[
0\le P_{4/3}\le P_{5/3}\le1,
\qquad
P_p=1-f_0^{2p}-f_1^{2p}.
\]
The local difference is then
\(c_{\mathrm D}P_{4/3}\rho_0^{4/3}-c_{\mathrm{TF}}P_{5/3}\rho_0^{5/3}\).
For \(0\le b\le a\) and \(t\ge0\), the elementary inequality
\[
c_{\mathrm D}b\,t^{4/3}-c_{\mathrm{TF}}a\,t^{5/3}
\le\frac{c_{\mathrm D}^2b^2}{4c_{\mathrm{TF}}a}\,t
\]
applied with \(a=P_{5/3}\), \(b=P_{4/3}\), \(t=\rho_0\) yields
\[
c_{\mathrm D}P_{4/3}\rho_0^{4/3}-c_{\mathrm{TF}}P_{5/3}\rho_0^{5/3}
\le\frac{c_{\mathrm D}^2P_{4/3}^2}{4c_{\mathrm{TF}}P_{5/3}}\rho_0
\le\alpha\rho_0
\]
whenever \(P_{5/3}>0\), while both local differences vanish when
\(P_{5/3}=0\).
The Coulomb cross term in~\eqref{eq:partition-localization}
has no factor \(1/2\), because it is \(2\mathcal{D}(f_0^2\rho_0,f_1^2\rho_0)\).

\begin{lemma}\label{lem:folded-localization-rough}
For every \(h>0\),
\begin{equation}\label{eq:folded-localization-rough}
    \frac{N^2}{16}
    \leq
    \frac{ZN}{4}
    +
    Ch\mathcal{D}(\rho_0,\rho_0)
    +
    C\left(
    h
    +
    \frac{1}{h}
    \right)N.
\end{equation}
\end{lemma}

\begin{proof}
Use normalized surface measure \(d\omega=dS/(4\pi)\) on \(\mathbb S^2\).
For \(\nu\in\mathbb S^2\), \(l>0\) and \(h>0\), define
\[
q(s)=\min\{1,\max\{0,s\}\},\qquad
a_0(s)=\cos\!\left(\frac\pi2q(s)\right),\quad
a_1(s)=\sin\!\left(\frac\pi2q(s)\right),
\]
and set
\[
p_\nu(x)=\nu\cdot F(x),\qquad
f_j(x)=a_j((p_\nu(x)-l)/h),\quad j=0,1.
\]
These functions are nonnegative, satisfy \(f_0^2+f_1^2=1\), and are
constant outside \(A_{\nu,l}=\{l<p_\nu<l+h\}\).
Since \(F\) is \(1\)-Lipschitz,
\[
|\nabla f_0|^2+|\nabla f_1|^2
 \le\frac{2(\pi/2)^2}{h^2}\mathbf1_{A_{\nu,l}}
 \qquad\text{a.e.}
\]
Thus~\eqref{eq:partition-localization} implies
\[
I_{\nu,l}
 \le \int_{\mathbb R^3} V f_1^2\rho_0\,dx+
 \left(\alpha+\frac{2c_{\mathrm W}(\pi/2)^2}{h^2}\right)
 \int_{A_{\nu,l}}\rho_0\,dx,
\]
where
\[
I_{\nu,l}=\iint_{\mathbb R^3\times\mathbb R^3}
 \frac{f_0(x)^2\rho_0(x)f_1(y)^2\rho_0(y)}{|x-y|}\,dx\,dy.
\]
For each \(x\), the set of \(l>0\) with \(x\in A_{\nu,l}\) has length
at most \(h\), whereas
\(\int_0^\infty f_1(x)^2\,dl\le(p_\nu(x))_+\).
The elementary spherical identity
\[
\int_{\mathbb S^2}(\nu\cdot z)_+\,d\omega(\nu)=\frac{|z|}{4}
\]
and \(\delta V\le Z\) therefore give
\[
\int_{\mathbb S^2}\int_0^\infty\int_{\mathbb R^3}V f_1^2\rho_0\,dx\,dl\,d\omega
 \le\frac14\int_{\mathbb R^3}\delta V\rho_0\,dx\le\frac{ZN}{4}.
\]
The averaged transition error is bounded by
\((\alpha h+2c_{\mathrm W}(\pi/2)^2/h)N\).

We spell out the cross-term lower bound to keep track of the factors.
For \(a,b\in\mathbb R\), set
\[
\begin{aligned}
L_h(a,b)&=\int_0^\infty
 a_0((a-l)/h)^2a_1((b-l)/h)^2\,dl,\\
H(a,b)&=(b-\max\{a,0\})_+.
\end{aligned}
\]
On \(\max\{a,0\}<l<b-h\) the integrand equals one, so
\(L_h(a,b)\ge H(a,b)-h\).
Moreover,
\[
H(a,b)+H(b,a)+H(-a,-b)+H(-b,-a)=|a-b|.
\]
Symmetrizing in \(x,y\) and in \(\nu,-\nu\), and using
\(\int_{\mathbb S^2}|\nu\cdot z|\,d\omega=|z|/2\), yields
\[
\begin{aligned}
\int_{\mathbb S^2}\int_0^\infty I_{\nu,l}\,dl\,d\omega
&\ge\frac18\iint_{\mathbb R^3\times\mathbb R^3}
 \frac{\rho_0(x)\rho_0(y)|F(x)-F(y)|}{|x-y|}\,dx\,dy
 -2h\mathcal{D}(\rho_0,\rho_0)\\
&\ge\frac18(M_1^2+M_2^2)-2h\mathcal{D}(\rho_0,\rho_0)\\
&\ge\frac{N^2}{16}-2h\mathcal{D}(\rho_0,\rho_0),
\end{aligned}
\]
where \(M_i=\int_{\Omega_i}\rho_0\,dx\), \(M_1+M_2=N\), and we used that
\(F\) is an isometry on each \(\Omega_i\).
By Tonelli's theorem all nonnegative integrals can be interchanged. Both
displayed bounds are integrable, the upper one because the cross-term average
is finite and the lower one because \(|F(x)-F(y)|\le|x-y|\) and
\(\mathcal{D}(\rho_0,\rho_0)<\infty\).
Combining the two bounds gives
\[
\frac{N^2}{16}\le\frac{ZN}{4}+2h\mathcal{D}(\rho_0,\rho_0)
 +\left(\frac{c_{\mathrm D}^2}{4c_{\mathrm{TF}}}h
       +\frac{2c_{\mathrm W}(\pi/2)^2}{h}\right)N,
\]
which proves~\eqref{eq:folded-localization-rough}.
\end{proof}

\begin{corollary}
    \begin{equation}\label{eq:rough-mass-detailed}
    N
    \leq
    4Z
    +
    CZ^{2/3}
    +
    C.
\end{equation}
\end{corollary}

\begin{proof}
    If \(N\le Z\), the claimed bound is immediate. Hence assume
    \(N>Z>0\), and write \(\mathcal{D}=\mathcal{D}(\rho_0,\rho_0)\).
    From~\eqref{eq:folded-localization-rough},
    \[
            \frac{N^2}{16}
            \leq
            \frac{ZN}{4}
            +
            Ch\mathcal{D}(\rho_0,\rho_0)
            +
            C\left(
            h
            +
            \frac{1}{h}
            \right)N
    \]

    Taking \(h=\sqrt{N/(\mathcal{D}+N)}\), we obtain
    \[
        \frac{N^2}{16}
        \leq
        \frac{ZN}{4}
        +
        C\sqrt{(\mathcal{D}+N)N}.
    \]
    Thus
    \[
        N
        \leq
        4Z
        +
        C\sqrt{\frac{\mathcal{D}+N}{N}}.
    \]
    From Proposition~\ref{prop:hartree-apriori}
    \[
        \mathcal{D}
        \leq
        C(Z^{7/3}+N),
    \]
    and \(N>Z\), we have
    \[
       \sqrt{\frac{\mathcal{D}+N}{N}}
       \le C\sqrt{\frac{Z^{7/3}}{N}+1}
       \le C(Z^{2/3}+1).
    \]
    This proves~\eqref{eq:rough-mass-detailed}.
\end{proof}

\section{Comparison with the Neutral Thomas--Fermi Density}
\label{sec:tf-comparison}

In this section we prove a quantitative comparison between the TFDW
minimizer and the neutral Thomas--Fermi minimizer.

\subsection{The neutral Thomas--Fermi reference density}

We consider the Thomas--Fermi functional
\[
\mathcal{E}_V^{\mathrm{TF}}(\rho)
=
c_{\mathrm{TF}}
\int_{\mathbb R^3}{\rho}^{5/3}
\,dx-
\int_{\mathbb R^3}V\rho
\,dx+
\mathcal{D}(\rho,\rho).
\]
We consider the global minimizer of this functional over all nonnegative
\(\rho\in L^1(\mathbb R^3)\cap L^{5/3}(\mathbb R^3)\). The neutral TF theorem
of Lieb and Simon~\cite{LiebSimon1977} gives a unique minimizer
\(\rho^{\mathrm{TF}}\), satisfying
\[
\int_{\mathbb R^3}\rho^{\mathrm{TF}}\,dx=Z.
\]

The screened potential of the neutral minimizer is the nonnegative function
on \(\mathbb R^3\) given by
\[
\Phi^{\mathrm{TF}}=V-U^{\mathrm{TF}} .
\]
Its Lagrange multiplier vanishes, and the neutral Euler equation
reads as in~\cite{LiebSimon1977}
\begin{equation}\label{eq:tf-euler}
\rho^{\mathrm{TF}}=\bigl(\frac{3}{5c_{\mathrm{TF}}}\Phi^{\mathrm{TF}}\bigr)_+^{3/2}
\quad\text{a.e.}.
\end{equation}

Moreover, the neutral Thomas--Fermi minimizer satisfies
\begin{equation}\label{eq:tf-standard-estimate}
    \int_{\mathbb R^3}(\rho^{\mathrm{TF}})^{5/3}\,dx\le CZ^{7/3}.
\end{equation}
Indeed, the argument for~\eqref{eq:nuclear-absorbed} applies to every
nonnegative \(L^1\cap L^{5/3}\) density. Applying it to
\(\rho^{\mathrm{TF}}\) and using
\(\mathcal{E}_V^{\mathrm{TF}}(\rho^{\mathrm{TF}})
\le\mathcal{E}_V^{\mathrm{TF}}(0)=0\) absorbs the attraction into the
positive TF and Coulomb terms and gives this bound, uniformly in the
nuclear positions.

\begin{lemma}\label{lem:tf-energy-gap}
Let \(\rho^{\mathrm{TF}}\) be the neutral TF minimizer.
Then for every nonnegative
\(\rho\in L^1(\mathbb R^3)\cap L^{5/3}(\mathbb R^3)\),
\[
\mathcal{D}(\rho-\rho^{\mathrm{TF}},\rho-\rho^{\mathrm{TF}})
\leq
\mathcal{E}_V^{\mathrm{TF}}(\rho)
-
\mathcal{E}_V^{\mathrm{TF}}(\rho^{\mathrm{TF}}).
\]
\end{lemma}

\begin{proof}
Expanding the Coulomb term,
\[
\mathcal{D}(\rho,\rho)
=
\mathcal{D}(\rho^{\mathrm{TF}},\rho^{\mathrm{TF}})
+2\mathcal{D}(\rho^{\mathrm{TF}},\rho-\rho^{\mathrm{TF}})
+\mathcal{D}(\rho-\rho^{\mathrm{TF}},\rho-\rho^{\mathrm{TF}}).
\]
Thus
\[
\begin{aligned}
&
\mathcal{E}_V^{\mathrm{TF}}(\rho)
-
\mathcal{E}_V^{\mathrm{TF}} (\rho^{\mathrm{TF}} )
-
\mathcal{D}(\rho-\rho^{\mathrm{TF}},\rho-\rho^{\mathrm{TF}})
\\
&=
c_{\mathrm{TF}}
\int_{\mathbb R^3}
\left(
\rho^{5/3}-(\rho^{\mathrm{TF}})^{5/3}
-\frac53 (\rho^{\mathrm{TF}})^{2/3}(\rho-\rho^{\mathrm{TF}})
\right)
\,dx\\
&\quad
+
\int_{\mathbb R^3}
\left(
\frac53c_{\mathrm{TF}}(\rho^{\mathrm{TF}})^{2/3}
-
V
+
U^{\mathrm{TF}}
\right)(\rho-\rho^{\mathrm{TF}})\,dx.
\end{aligned}
\]
The first term is nonnegative by convexity. Since \(\rho^{\mathrm{TF}}\) is
strictly positive in~\cite[Theorem~IV.3]{LiebSimon1977}, the Euler
equation~\eqref{eq:tf-euler} gives
\[
\frac53c_{\mathrm{TF}}(\rho^{\mathrm{TF}})^{2/3}-V+U^{\mathrm{TF}}=0
\qquad \text{a.e.},
\]
whence
\[
\int_{\mathbb R^3}
\left(\frac53c_{\mathrm{TF}}(\rho^{\mathrm{TF}})^{2/3}-V+U^{\mathrm{TF}}\right)
(\rho-\rho^{\mathrm{TF}})\,dx=0 .
\]
This proves the claim.
\end{proof}

Let \(\rho_i^{\mathrm{TF}}\) be the neutral TF minimizer for the single
nucleus of charge \(Z_i\) at \(R_i\), and let \(U_i^{\mathrm{TF}}\) denote
its Coulomb potential,
\[
    U_i^{\mathrm{TF}}(x)
    =
    \int_{\mathbb R^3}\frac{\rho_i^{\mathrm{TF}}(y)}{|x-y|}\,dy .
\]
We write \(\Phi_i=Z_i|x-R_i|^{-1}-U_i^{\mathrm{TF}}\).

For this minimizer,
\[
    0 < \Phi_i(x)\le A|x-R_i|^{-4},
    \qquad
    \rho_i^{\mathrm{TF}}(x)=\gamma\Phi_i(x)^{3/2},
    \qquad x\ne R_i ,
\]
where \(A=(5c_{\mathrm{TF}})^3/(3\pi^2)\) and
\(\gamma=(3/5c_{\mathrm{TF}})^{3/2}\). The bound on \(\Phi_i\) is
in~\cite[Theorem~4.2]{FrankNamVanDenBosch2018}.
Moreover, \(\Phi_i\) and \(\Phi^{\mathrm{TF}}\) are smooth away from the
nuclei in~\cite[Theorem~IV.5(b)]{LiebSimon1977}.

\begin{lemma}\label{lem:tf-pointwise-tail}
\[
\rho^{\mathrm{TF}}(x)
\leq
\sqrt{2}
[\rho_1^{\mathrm{TF}}(x)
+
\rho_2^{\mathrm{TF}}(x)]
\]
for every \(x\neq R_1,R_2\).
Moreover,
\[
0\leq\Phi^{\mathrm{TF}}(x)
\leq
\Phi_1(x)
+
\Phi_2(x).
\]
\end{lemma}

\begin{proof}
Lieb and Simon proved that \(\Phi^{\mathrm{TF}}\ge0\)
in~\cite[Theorem~IV.2(a)]{LiebSimon1977}.
Write \(w=\Phi^{\mathrm{TF}}-\Phi_1-\Phi_2\) and set
    \[
        \mathcal B=\{x\in\mathbb R^3\setminus\{R_1,R_2\}:w(x)>0\}.
    \]
Since \(w\) is continuous, \(\mathcal B\) is open.
Assume that \(\mathcal B\) is nonempty and fix \(x_0\in \mathcal B\). Then
\[
\Phi^{\mathrm{TF}}(x_0)>\Phi_1(x_0)+\Phi_2(x_0).
\]

Since \(\Phi^{\mathrm{TF}}(x_0)>\Phi_1(x_0)+\Phi_2(x_0)>0\), the Euler
equation gives at \(x_0\)
\[
    \begin{aligned}
        \rho^{\mathrm{TF}}(x_0)^{2/3}
&=
\frac{3}{5c_{\mathrm{TF}}}\Phi^{\mathrm{TF}}(x_0)
\\
&>
\frac{3}{5c_{\mathrm{TF}}}
\bigl(\Phi_1(x_0)+\Phi_2(x_0)\bigr)
\\
&=
\rho_1^{\mathrm{TF}}(x_0)^{2/3}
+
\rho_2^{\mathrm{TF}}(x_0)^{2/3}.
    \end{aligned}
\]
Since \(t\mapsto t^{3/2}\) is strictly increasing on \([0,\infty)\),
\[
\begin{aligned}
    \rho^{\mathrm{TF}}(x_0)
&>
\left(
\rho_1^{\mathrm{TF}}(x_0)^{2/3}
+
\rho_2^{\mathrm{TF}}(x_0)^{2/3}
\right)^{3/2}
\\
&>
\rho_1^{\mathrm{TF}}(x_0)
+
\rho_2^{\mathrm{TF}}(x_0).
\end{aligned}
\]
As \(x_0\in \mathcal B\) is arbitrary,
\[
\rho^{\mathrm{TF}}
>
\rho_1^{\mathrm{TF}}
+
\rho_2^{\mathrm{TF}}
\qquad\text{in } \mathcal B.
\]

Direct computation yields
\[
\Delta\Phi^{\mathrm{TF}}
=
4\pi\rho^{\mathrm{TF}}
\qquad
\Delta\Phi_i
=
4\pi\rho_i^{\mathrm{TF}}
\qquad \text{in } \mathbb{R}^3\setminus\{R_1, R_2\},
\]
which gives
\[
\Delta w
=
4\pi
\bigl(
\rho^{\mathrm{TF}}
-
\rho_1^{\mathrm{TF}}
-
\rho_2^{\mathrm{TF}}
\bigr).
\]
Hence,
\begin{equation}\label{eq:w-strict-subharmonic}
    \Delta w>0
\qquad\text{in }\mathcal B
\end{equation}

Near \(R_i\), we have
\[
\Phi^{\mathrm{TF}}(x)=\frac{Z_i}{|x-R_i|}+O(1),
\qquad
\Phi_i(x)=\frac{Z_i}{|x-R_i|}+O(1),
\]
since \(\rho^{\mathrm{TF}}\) and \(\rho_i^{\mathrm{TF}}\) are integrable near
each
nucleus. For \(j\neq i\), the potential \(\Phi_j\) is bounded near \(R_i\). Thus
\(w\) remains bounded near \(R_i\).
Moreover,
\[
w(x)\to0
\qquad\text{as }|x|\to\infty.
\]

Choose \(\varepsilon>0\) small enough so that \(w_\varepsilon(x_0)>0\), where
\[
w_\varepsilon(x)=w(x)-\varepsilon
\left(|x-R_1|^{-1}+|x-R_2|^{-1}\right).
\]
The boundedness of \(w\) near each nucleus implies that
\(w_\varepsilon\to-\infty\) there, while
\(w_\varepsilon\to0\) at infinity. Hence \(w_\varepsilon\) attains a
positive maximum at a point \(x_*\) away from the nuclei.
At such a point \(w(x_*)>0\), and the subtracted term is harmonic, so
\(\Delta w_\varepsilon(x_*)=\Delta w(x_*)\).
Then~\eqref{eq:w-strict-subharmonic} gives
\[
0\ge\Delta w_\varepsilon(x_*)=\Delta w(x_*)>0.
\]
This is a contradiction. Hence
\(\mathcal B=\varnothing\), and therefore
\[
0\leq
\Phi^{\mathrm{TF}}
\leq
\Phi_1+\Phi_2.
\]
Then we have
\[
    \begin{aligned}
        \rho^{\mathrm{TF}}(x)^{2/3}
&=
\frac{3}{5c_{\mathrm{TF}}}\Phi^{\mathrm{TF}}(x)
\\
&\leq
\frac{3}{5c_{\mathrm{TF}}}
\bigl(
\Phi_1(x)+\Phi_2(x)
\bigr)
\\
&=
\rho_1^{\mathrm{TF}}(x)^{2/3}
+
\rho_2^{\mathrm{TF}}(x)^{2/3}.
    \end{aligned}
\]

Thus
\begin{equation}
\rho^{\mathrm{TF}}(x)
\leq
\left(
\rho_1^{\mathrm{TF}}(x)^{2/3}
+
\rho_2^{\mathrm{TF}}(x)^{2/3}
\right)^{3/2}.
\label{eq:density-intermediate}
\end{equation}
Since \((a+b)^{3/2}\le\sqrt2\,(a^{3/2}+b^{3/2})\) for \(a,b\ge0\), we obtain
\[
\rho^{\mathrm{TF}}(x)
\leq
\sqrt2\,
\bigl(
\rho_1^{\mathrm{TF}}(x)
+
\rho_2^{\mathrm{TF}}(x)
\bigr)
\]
for every
\(x\in\mathbb R^3\setminus\{R_1,R_2\}\).
\end{proof}

We record the following decay estimates for the neutral TF density and
potential.

\begin{corollary}\label{cor:tf-tail}
There exists \(C>0\) such that the following hold for every
\(x\in\mathbb R^3\setminus\{R_1,R_2\}\):
\[
\begin{aligned}
    0\leq \rho^{\mathrm{TF}}(x)
\leq&
C\left(
|x-R_1|^{-6}
+
|x-R_2|^{-6}
\right),
\\
0\leq\Phi^{\mathrm{TF}}(x)
\leq&
C\left(
|x-R_1|^{-4}
+
|x-R_2|^{-4}
\right).
\end{aligned}
\]
Moreover, for every \(r>0\) and \(0<\lambda\leq1\),
\[
\int_{\{\delta\geq r\}}
\rho^{\mathrm{TF}}
\,dx\leq
Cr^{-3},
\]
and
\[
\int_{K_{(1+\lambda)r}\setminus K_r}
\rho^{\mathrm{TF}}
\,dx\leq
C\lambda r^{-3},
\]
and
\[
    0
\leq
\delta(x)\Phi_r^{\mathrm{TF}}(x)
\leq
Cr^{-3},
\qquad
\delta(x)\geq r.
\]
\end{corollary}

\begin{proof}
The pointwise bounds follow from Lemma~\ref{lem:tf-pointwise-tail}
and the atomic Sommerfeld bound. Since \(\delta(x)\ge r\) implies
\(|x-R_i|\ge r\) for both nuclei, integration gives
\[
\int_{\{\delta\ge r\}}\rho^{\mathrm{TF}}
 \,dx\le C\sum_{i=1}^2\int_{|x-R_i|\ge r}|x\,dx-R_i|^{-6}\,dx
 \le Cr^{-3}.
\]
The shell \(K_{(1+\lambda)r}\setminus K_r\) is contained in the union
of the two spherical annuli with radii \(r\) and \((1+\lambda)r\).
Its volume is at most \(C\lambda r^3\), and
\(\rho^{\mathrm{TF}}\le Cr^{-6}\) there. This proves the shell bound.

For the screened potential, put
\[
U_{\geq r}^{\mathrm{TF}}(x)
=\int_{\{\delta>r\}}\frac{\rho^{\mathrm{TF}}(y)}{|x-y|}\,dy .
\]
Then
\[
\Phi_r^{\mathrm{TF}}
 =\Phi^{\mathrm{TF}}+U_{\geq r}^{\mathrm{TF}},\qquad
0\le e_r^{\mathrm{TF}}(y)
 \le C\sum_{i=1}^2
 \mathbf1_{\{|y-R_i|>r\}}|y-R_i|^{-6}.
\]
For \(t_i=|x-R_i|\ge r\), Newton's theorem applied to each radial
majorant yields
\[
\begin{aligned}
\int_{|y-R_i|>r}\frac{|y-R_i|^{-6}}{|x-y|}\,dy
 &=4\pi\left(t_i^{-1}\int_r^{t_i}s^{-4}\,ds
                   +\int_{t_i}^{\infty}s^{-5}\,ds\right)\\
 &\le\frac{4\pi}{t_i}\int_r^\infty s^{-4}\,ds
 =\frac{4\pi}{3t_i}r^{-3}.
\end{aligned}
\]
Because \(\delta(x)\le t_i\), we obtain
\(\delta(x)U_{\geq r}^{\mathrm{TF}}(x)\le Cr^{-3}\).
The pointwise potential bound also gives
\(\delta(x)\Phi^{\mathrm{TF}}(x)
\le C\sum_i t_i^{-3}\le Cr^{-3}\).
Both terms are nonnegative, proving the last assertion.
\end{proof}

\subsection{Comparison with the TF minimizer}

We first smooth the Thomas--Fermi density at scale \(s\) to obtain a trial
density with \(\int_{\mathbb R^3}\sigma_s\,dx=Z\).

\begin{lemma}\label{lem:tf-smoothing}
For every \(s>0\), there exists a smooth nonnegative density \(\sigma_s\)
satisfying
\[
\begin{aligned}
    \int_{\mathbb R^3}\sigma_s\,dx&=Z,
    &\qquad
    \int_{\mathbb R^3}\sigma_s^{5/3}\,dx&\leq\int_{\mathbb R^3}(\rho^{\mathrm{TF}})^{5/3}\,dx,
    \\
    \mathcal{D}(\sigma_s,\sigma_s)&\leq\mathcal{D}(\rho^{\mathrm{TF}},\rho^{\mathrm{TF}}),
    &\qquad
    \int_{\mathbb R^3}|\nabla\sqrt{\sigma_s}|^2\,dx&\leq Cs^{-2}Z,
\end{aligned}
\]
and
\[
    0\leq\int_{\mathbb R^3}V(\rho^{\mathrm{TF}}-\sigma_s)\,dx\leq CZ^{12/5}s^{1/5}.
\]
\end{lemma}

\begin{proof}
Choose a nonnegative radial \(g\in C_c^\infty(B_1)\) with
\(\int_{\mathbb R^3} g^2\,dx=1\), and set \(\eta=g^2\). Thus \(\eta\) is a smooth
radial probability kernel with smooth square root and support in
\(B_1\). Set
\[
\eta_s(x)=s^{-3}\eta(x/s),
\qquad
\sigma_s=\rho^{\mathrm{TF}}*\eta_s.
\]
Then we define
\[
    g_s=\sqrt{\eta_s}, \qquad
    Q_s(x) = \int_{\mathbb{R}^3} \rho^{\mathrm{TF}}(y) |\nabla g_s(x-y)|^2 \, dy.
\]
Tonelli's theorem gives
\[
    \int_{\mathbb R^3}\sigma_s\,dx=Z,
\]
and Jensen's
inequality gives
\[
\int_{\mathbb R^3}\sigma_s^{5/3}\,dx\leq\int_{\mathbb R^3} (\rho^{\mathrm{TF}})^{5/3}\,dx.
\]
Since \(|\widehat{\eta_s}|\leq1\), the Fourier representation of
the Coulomb energy gives
\[
\mathcal{D}(\sigma_s,\sigma_s)\leq\mathcal{D}(\rho^{\mathrm{TF}},\rho^{\mathrm{TF}}).
\]

Applying weighted Cauchy--Schwarz to the derivatives
of the convolution yields
\[
    \begin{aligned}
|\nabla \sigma_s(x)|^2 &\le 4 \left( \int_{\mathbb{R}^3} \rho^{\mathrm{TF}}(y) g_s(x-y)^2 \, dy \right) \left( \int_{\mathbb{R}^3} \rho^{\mathrm{TF}}(y) |\nabla g_s(x-y)|^2 \, dy \right) \\
&= 4 \sigma_s(x) Q_s(x),
\end{aligned}
\]
and the translation invariance gives
\[
    \begin{aligned}
\int_{\mathbb{R}^3} Q_s(x) \, dx
&= \int_{\mathbb{R}^3} \rho^{\mathrm{TF}}(y)
\left( \int_{\mathbb{R}^3} |\nabla g_s(x-y)|^2 \, dx \right) dy \\
&= \left( \int_{\mathbb R^3} \rho^{\mathrm{TF}}(y) \, dy \right) \left( \int_{\mathbb R^3} |\nabla g_s(z)|^2 \, dz \right)
\\
&= Z \|\nabla g_s\|_2^2 \\
&= Z s^{-2} \|\nabla g\|_2^2.
\end{aligned}
\]
Hence we have
\[
|\nabla\sigma_s|^2
\leq
4\sigma_s Q_s,
\qquad
\int_{\mathbb R^3} Q_s\,dx\leq Cs^{-2}Z.
\]
Indeed, the functions
\(u_\varepsilon=\sqrt{\sigma_s+\varepsilon}-\sqrt\varepsilon\)
satisfy \(u_\varepsilon^2\le\sigma_s\) and
\(|\nabla u_\varepsilon|^2\le Q_s\). They converge to
\(\sqrt{\sigma_s}\) in \(L^2\) by dominated convergence.
Weak compactness in \(H^1\) and lower semicontinuity therefore give
\[
\sqrt{\sigma_s}\in H^1(\mathbb R^3),
\qquad
\int_{\mathbb R^3}|\nabla\sqrt{\sigma_s}|^2
\,dx\leq
Cs^{-2}Z.
\]

Finally, we define
\[
\kappa_s
=
\frac{1}{|\cdot|}
-
\frac{1}{|\cdot|}*\eta_s .
\]
Since \(\sigma_s=\rho^{\mathrm{TF}}*\eta_s\), Tonelli's theorem and the radial
symmetry of \(\eta_s\) give
\[
\begin{aligned}
\int_{\mathbb R^3}
V(\rho^{\mathrm{TF}}-\sigma_s)
\,dx&=
\int_{\mathbb R^3}
\rho^{\mathrm{TF}}(x)
\left(
V(x)-V*\eta_s(x)
\right)\,dx\\
&=
\int_{\mathbb R^3}
\rho^{\mathrm{TF}}(x)
\left(
Z_1\kappa_s(x-R_1)+Z_2\kappa_s(x-R_2)
\right)dx \,dx.
\end{aligned}
\]
Newton's theorem and \(\operatorname{supp}(\eta_s)\subset B_s\) give
\[
0\le \kappa_s(x)\le |x|^{-1}\mathbf1_{\{|x|<s\}}.
\]
Consequently,
\[
\|\kappa_s\|_{L^{5/2}}
 \le\left(4\pi\int_0^s t^{-1/2}\,dt\right)^{2/5}
 \le Cs^{1/5}.
\]
H\"older's inequality now gives
\[
\begin{aligned}
0
\le
\int_{\mathbb R^3}
V(\rho^{\mathrm{TF}}-\sigma_s)
\,dx&\le
CZ
\|\rho^{\mathrm{TF}}\|_{L^{5/3}}
\|\kappa_s\|_{L^{5/2}}
\\
&\leq
CZ
\|\rho^{\mathrm{TF}}\|_{L^{5/3}}
s^{1/5}
\\
&\leq
CZ^{12/5}s^{1/5}.
\end{aligned}
\]
The last inequality follows from~\eqref{eq:tf-standard-estimate}.
\end{proof}

\begin{proposition}
\label{prop:tf-comparison}
Assume \(N\ge Z\ge 2\). Then
\[
\mathcal{D}(\rho_0-\rho^{\mathrm{TF}},\rho_0-\rho^{\mathrm{TF}})
\le C Z^{25/11},
\]
and
\[
\|\rho_0-\rho^{\mathrm{TF}}\|_{L^{5/3}(\mathbb R^3)}
\le C Z^{7/5}.
\]
\end{proposition}

\begin{proof}
Let \(\sigma_s\) be the density from Lemma~\ref{lem:tf-smoothing}. Since
\(\sqrt{\sigma_s}\in H^1(\mathbb R^3)\) and \(\int_{\mathbb R^3}\sigma_s\,dx=Z\),
it is an admissible trial density for the variational problem with
\(\int_{\mathbb R^3}\sigma\,dx=Z\), so
\[
E_V^{\mathrm{TFDW}}(Z)
\le \mathcal{E}_V^{\mathrm{TFDW}}(\sigma_s).
\]
Then we have
\[
E_V^{\mathrm{TFDW}}(Z)
\le
\mathcal{E}_V^{\mathrm{TF}}(\sigma_s)
+
C Zs^{-2}.
\]

Since
\[
\int_{\mathbb R^3}\sigma_s^{5/3}
\,dx\le
\int_{\mathbb R^3}(\rho^{\mathrm{TF}})^{5/3}\,dx,
\]
and
\[
\mathcal{D}(\sigma_s,\sigma_s)
\le
\mathcal{D}(\rho^{\mathrm{TF}},\rho^{\mathrm{TF}}),
\]
we have
\[
\begin{aligned}
\mathcal{E}_V^{\mathrm{TF}}(\sigma_s)
-
\mathcal{E}_V^{\mathrm{TF}}(\rho^{\mathrm{TF}})
&\le
\int_{\mathbb R^3}
V(\rho^{\mathrm{TF}}-\sigma_s)
\,dx\\
&\le
C Z^{12/5}s^{1/5}.
\end{aligned}
\]
Consequently,
\[
E_V^{\mathrm{TFDW}}(Z)
-
\mathcal{E}_V^{\mathrm{TF}}(\rho^{\mathrm{TF}})
\le
C\left(
Zs^{-2}
+
Z^{12/5}s^{1/5}
\right).
\]

Taking \(s=Z^{-7/11}\), we obtain
\[
E_V^{\mathrm{TFDW}}(Z)
-
\mathcal{E}_V^{\mathrm{TF}}(\rho^{\mathrm{TF}})
\le
C Z^{25/11}.
\]

Since \(N\ge Z\), the monotonicity in the electron number proved
after~\eqref{eq:binding-inequality} gives
\[
E_V^{\mathrm{TFDW}}(N)
\le
E_V^{\mathrm{TFDW}}(Z).
\]
Namely,
\[
\mathcal{E}_V^{\mathrm{TFDW}}(\rho_0)
-
\mathcal{E}_V^{\mathrm{TF}}(\rho^{\mathrm{TF}})
\le
C Z^{25/11}.
\]

The definition of the TFDW functional gives
\[
\mathcal{E}_V^{\mathrm{TF}}(\rho_0)
=
\mathcal{E}_V^{\mathrm{TFDW}}(\rho_0)
-
c_{\mathrm W}
\int_{\mathbb R^3}
|\nabla\sqrt{\rho_0}|^2
\,dx+
c_{\mathrm D}
\int_{\mathbb R^3}{\rho_0}^{4/3}\,dx,
\]
and hence
\[
\mathcal{E}_V^{\mathrm{TF}}(\rho_0)
\le
\mathcal{E}_V^{\mathrm{TFDW}}(\rho_0)
+
c_{\mathrm D}
\int_{\mathbb R^3}{\rho_0}^{4/3}\,dx.
\]
From Cauchy--Schwarz, we have
\[
\begin{aligned}
\int_{\mathbb R^3}{\rho_0}^{4/3}\,dx
&=\int_{\mathbb R^3}{\rho_0}^{1/2}{\rho_0}^{5/6}\,dx
\le N^{1/2}\left(\int_{\mathbb R^3}{\rho_0}^{5/3}\,dx\right)^{1/2}\\
&\le CZ^{5/3}\le CZ^{25/11}.
\end{aligned}
\]
Here we used the previous estimates
\[
N\le CZ,
\qquad
\int_{\mathbb R^3}{\rho_0}^{5/3}
\,dx\le CZ^{7/3}.
\]
Therefore, we have
\[
\mathcal{E}_V^{\mathrm{TF}}(\rho_0)
-
\mathcal{E}_V^{\mathrm{TF}}(\rho^{\mathrm{TF}})
\le
CZ^{25/11}.
\]

Applying Lemma~\ref{lem:tf-energy-gap}, we obtain
\[
\mathcal{D}(\rho_0-\rho^{\mathrm{TF}},\rho_0-\rho^{\mathrm{TF}})
\le
\mathcal{E}_V^{\mathrm{TF}}(\rho_0)
-
\mathcal{E}_V^{\mathrm{TF}}(\rho^{\mathrm{TF}})
\le
CZ^{25/11}.
\]

The previous a priori bound gives
\[
\int_{\mathbb R^3}{\rho_0}^{5/3}
\,dx\le
CZ^{7/3},
\]
while the standard Thomas--Fermi estimate~\cite{LiebSimon1977} gives
\[
\int_{\mathbb R^3}
(\rho^{\mathrm{TF}})^{5/3}
\,dx\le
CZ^{7/3}.
\]
Therefore,
\[
\|\rho_0-\rho^{\mathrm{TF}}\|_{L^{5/3}}
\le
\|{\rho_0}\|_{L^{5/3}}
+
\|\rho^{\mathrm{TF}}\|_{L^{5/3}}
\le
CZ^{7/5}.
\]
This completes the proof.
\end{proof}

\section{Screened Coulomb Potentials}
\label{sec:coulomb-truncation}

In this section we prove the interior estimate for the Coulomb potential. We
bound the potential of a density restricted to \(K_r\) pointwise in terms of
the Coulomb energy and the \(L^{5/3}\) norm of the density, uniformly in the
relative position of the two nuclei.

Here \(\dot H^1\) denotes the homogeneous Sobolev space, with norm
\(\|\nabla\cdot\|_2\), and \(\|\cdot\|_{(\dot H^1)^*}\) the norm of its dual.
For a
signed density the physical Coulomb integral need not converge. The following
proposition identifies it with the squared dual norm whenever it does, and this
will serve as the definition of the Coulomb energy for signed densities.

\begin{proposition}\label{prop:coulomb-extension}
Let \(\zeta\in L^{5/3}(\mathbb R^3)\) satisfy
\(\iint_{\mathbb R^3\times\mathbb R^3}\frac{|\zeta(x)\zeta(y)|}{|x-y|}\,dx\,dy<\infty\). Then
\[
\frac12\iint_{\mathbb R^3\times\mathbb R^3}\frac{\zeta(x)\zeta(y)}{|x-y|}\,dx\,dy
=
2\pi\|\zeta\|_{(\dot H^1)^*}^2 .
\]
\end{proposition}

\begin{proof}
For smooth compactly supported \(\zeta\), Plancherel's theorem and
\(\widehat{|\cdot|^{-1}}=4\pi|\xi|^{-2}\) give
\[
2\pi\|\zeta\|_{(\dot H^1)^*}^2
=
\frac{1}{4\pi^2}\int_{\mathbb R^3}\frac{|\hat\zeta(\xi)|^2}{|\xi|^2}\,d\xi
=
\frac12\iint_{\mathbb R^3\times\mathbb R^3}\frac{\zeta(x)\zeta(y)}{|x-y|}\,dx\,dy .
\]
Both sides are continuous in the \(L^{6/5}\) norm, because
\[
\Big|\int_{\mathbb R^3}\zeta\varphi\,dx\Big|
\le
C\|\zeta\|_{6/5}\|\nabla\varphi\|_2
\qquad\text{for }\varphi\in\dot H^1,
\]
and
\[
\frac12\iint_{\mathbb R^3\times\mathbb R^3}\frac{|\zeta(x)\zeta(y)|}{|x-y|}\,dx\,dy
\le
C\|\zeta\|_{6/5}^2 ,
\]
by H\"older's inequality, the Sobolev embedding \(\dot H^1\subset L^6\), and the
Hardy--Littlewood--Sobolev inequality.

Since \(C_c^\infty(\mathbb R^3)\) is dense in \(L^{6/5}\) and
\(\zeta_L=\mathbf1_{B_L(0)}\zeta\in L^{6/5}(\mathbb R^3)\), the identity holds
with
\(\zeta\) replaced by \(\zeta_L\) or by \(\zeta_L-\zeta_M\). If the physical
integral converges absolutely, dominated convergence gives
\[
\|\zeta_L-\zeta_M\|_{(\dot H^1)^*}\longrightarrow0
\qquad(L,M\to\infty),
\]
so \((\zeta_L)\) is a Cauchy sequence in \((\dot H^1)^*\), and
\[
\int_{\mathbb R^3}\zeta_L\varphi\,dx\longrightarrow\int_{\mathbb R^3}\zeta\varphi
\,dx\qquad\bigl(\varphi\in C_c^\infty(\mathbb R^3)\bigr)
\]
identifies its limit with \(\zeta\). A final application of dominated
convergence
to the identity with \(\zeta=\zeta_L\) shows
\[
\frac12\iint_{\mathbb R^3\times\mathbb R^3}\frac{\zeta(x)\zeta(y)}{|x-y|}\,dx\,dy
=
2\pi\|\zeta\|_{(\dot H^1)^*}^2 .
\]
\end{proof}

\begin{definition}\label{def:coulomb-energy}
For a signed density \(\zeta\) with \(\|\zeta\|_{(\dot H^1)^*}<\infty\) we
define
\[
\mathcal{D}(\zeta,\zeta)
:=
2\pi\|\zeta\|_{(\dot H^1)^*}^2
=
2\pi\sup_{\substack{\varphi\in C_c^\infty(\mathbb R^3)\\
                    \|\nabla\varphi\|_2\leq1}}
\Big|\int_{\mathbb R^3}\zeta\varphi\,dx\Big|^2 ,
\]
the supremum being attained on \(\dot H^1\). Equivalently, by Plancherel's
theorem,
\[
\mathcal{D}(\zeta,\zeta)
=
\frac{1}{4\pi^2}\int_{\mathbb R^3}\frac{|\hat\zeta(\xi)|^2}{|\xi|^2}\,d\xi .
\]
Finite energy means that \(\zeta\) extends continuously to \(\dot H^1\), and by
Proposition~\ref{prop:coulomb-extension} the definition agrees with the physical
Coulomb integral whenever the latter converges absolutely. In particular, no
global
\(L^1\) assumption is required in the theorem below.
\end{definition}

For every test function \(\psi\in\dot H^1\), the estimate corresponding
to~\eqref{eq:fourier-parseval} reads
\begin{equation}\label{eq:coulomb-duality}
\left|
\int_{\mathbb R^3}\zeta\psi
\,dx\right|
\leq
C \mathcal{D}(\zeta,\zeta)^{1/2}
\|\nabla\psi\|_{L^2}.
\end{equation}
If \(\zeta\in L^{5/3}\) and \(\psi\) is compactly supported and Lipschitz,
the integral above is an ordinary Lebesgue integral, since mollifications of
\(\psi\) converge to \(\psi\) in \(L^{5/2}\) and the integrals against
\(\zeta\) converge by H\"older's inequality.

\begin{lemma}\label{lem:shell-volume}
For \(0<t<r\) and \(i=1,2\), let
\[
S_i=\{x:r<|x-R_i|<r+t\} .
\]
Then there is a universal constant \(C\) such that, uniformly in
\(x\in\mathbb R^3\), \(s>0\), and \(i\in\{1,2\}\),
\[
|S_i\cap B_s(x)|\leq Cts^2
\qquad\text{and}\qquad
|S_i|\leq Cr^2t .
\]
\end{lemma}

\begin{proof}
Fix \(x\in\mathbb R^3\), \(s>0\), and a center \(a=R_i\), and set
\(d=|x-a|\). We claim that for every \(\tau>0\)
\[
\mathcal H^2\bigl(\partial B_\tau(a)\cap B_s(x)\bigr)\le16\pi s^2 .
\]
Indeed, if \(s\geq\tau/2\), then the sphere has total area \(4\pi\tau^2\), so
\[
\mathcal H^2\bigl(\partial B_\tau(a)\cap B_s(x)\bigr)\le4\pi\tau^2\le16\pi s^2 .
\]
Otherwise a nonempty intersection has \(d>\tau/2\), and its area is the
spherical
cap area
\[
\frac{\pi\tau}{d}\bigl(s^2-(\tau-d)^2\bigr)\le2\pi s^2 .
\]
The empty intersection contributes nothing.

Polar integration over \(r<\tau<r+t\) therefore gives
\[
|S_i\cap B_s(x)|
=
\int_r^{r+t}\mathcal H^2\,dx\bigl(\partial B_\tau(a)\cap B_s(x)\bigr)\,d\tau
\le16\pi s^2t .
\]
The same integration, with \(\tau<2r\), gives
\[
|S_i|=4\pi\int_r^{r+t}\tau^2\,d\tau\le16\pi r^2t .
\]
\end{proof}

\begin{theorem}\label{thm:screened-coulomb}
Let \(\zeta\in L^{5/3}(\mathbb R^3)\) with \(\mathcal{D}(\zeta,\zeta)<\infty\).
Then, for every \(r>0\), \(0<t<r/4\), and \(x\in\mathbb R^3\),
\[
\Big|\int_{K_r}\frac{\zeta(y)}{|x-y|}\,dy\Big|
\leq
C\mathcal{D}(\zeta,\zeta)^{1/2}t^{-1/2}\bigl[1+\log(r/t)\bigr]^{1/2}
+
C\|\zeta\|_{5/3}t^{1/5}.
\]
\end{theorem}

\begin{proof}
For brevity we write
\[
k_x(y)=\frac{\mathbf 1_{K_r}(y)}{|x-y|} ,
\]
so that \(\int_{K_r}\frac{\zeta(y)}{|x-y|}\,dy=\int_{\mathbb R^3}\zeta k_x\,dx\). Since
\(k_x\notin\dot H^1\), the estimate~\eqref{eq:coulomb-duality} does not apply
to it directly. We therefore set \(\psi=\theta v_t\), where
\[
v_t(y)=\frac1{\max\{|x-y|,t\}},
\qquad
\theta(y)=\min\left\{1,\max\left\{0,\frac{r+t-\delta(y)}{t}\right\}\right\} .
\]
The two terms \(\int_{\mathbb R^3}\zeta(k_x-\psi)\,dx\) and \(\int_{\mathbb R^3}\zeta\psi\,dx\) are then
estimated separately.

By construction
\[
\theta=1\text{ in }K_r,
\qquad
\theta=0\text{ outside }K_{r+t},
\qquad
0\le\theta\le1 ,
\]
with the transition layer contained in \(S_1\cup S_2\). In particular \(\psi\)
has compact support in \(\overline{K_{r+t}}\). Since
\(\delta\) is \(1\)-Lipschitz,
\[
|\nabla\theta|\le t^{-1},
\qquad
v_t\le t^{-1},
\qquad
|\nabla v_t|\le t^{-2}
\qquad\text{a.e.},
\]
so \(\psi\) is Lipschitz as well. Hence \(\psi\in\dot H^1\cap L^{5/2}\) is an
admissible test in~\eqref{eq:coulomb-duality}. Moreover \(k_x\in L^{5/2}\),
because
\[
\int_{K_r}\frac{dy}{|x-y|^{5/2}}
\le
\int_{B_{2r}(x)}\frac{dz}{|z|^{5/2}}
+
|K_r|\,(2r)^{-5/2}
\le
Cr^{1/2}
<\infty .
\]
All pairings below are therefore absolutely convergent.

We first estimate \(k_x-\psi\):
\[
\begin{aligned}
|k_x(y)-\psi(y)|
&=
\Bigl|\frac{\mathbf 1_{K_r}(y)}{|x-y|}-\theta(y)v_t(y)\Bigr|\\
&\le
\mathbf 1_{K_r}(y)\Bigl|\frac1{|x-y|}-v_t(y)\Bigr|
+
|\mathbf 1_{K_r}(y)-\theta(y)|\,v_t(y)\\
&=
\mathbf 1_{K_r}(y)\Bigl(\frac1{|x-y|}-\frac1t\Bigr)_+
+
|\mathbf 1_{K_r}(y)-\theta(y)|\,v_t(y)\\
&\le
\Bigl(\frac1{|x-y|}-\frac1t\Bigr)_+
+
\mathbf 1_{S_1}(y)v_t(y)
+
\mathbf 1_{S_2}(y)v_t(y).
\end{aligned}
\]
The first term of the last line is supported in \(B_t(x)\), and the substitution
\(y=x+tz\) gives
\[
\int_{\mathbb R^3}\,dx\Bigl(\frac1{|x-y|}-\frac1t\Bigr)_+^{5/2}dy
=
t^{1/2}\int_{B_1}\,dx\Bigl(\frac1{|z|}-1\Bigr)_+^{5/2}dz
\le
Ct^{1/2} .
\]
For the other two terms, set \(F(s)=|S_i\cap B_s(x)|\). Tonelli's theorem gives,
for every \(p>0\),
\begin{equation}\label{eq:vt-moment}
\int_{S_i}v_t^p\,dx=p\int_t^\infty F(s)s^{-p-1}\,ds .
\end{equation}
Taking \(p=5/2\) and using the bound \(F(s)\le Cts^2\) of
Lemma~\ref{lem:shell-volume}, we have
\[
\int_{S_i}v_t^{5/2}
\,dx=
\frac52\int_t^\infty F(s)s^{-7/2}\,ds
\le
Ct\int_t^\infty s^{-3/2}\,ds
\le
Ct^{1/2} .
\]
Each of the three terms is therefore bounded in \(L^{5/2}\) by \(Ct^{1/2}\), so
\begin{equation}\label{eq:trunc-error}
\|k_x-\psi\|_{5/2}\le Ct^{1/5}.
\end{equation}

Next, we have
\[
\begin{aligned}
\int_{\mathbb R^3}|\nabla\psi|^2
\,dx&=
\int_{\mathbb R^3}|\theta\nabla v_t\,dx+v_t\nabla\theta|^2\\
&\le
2\int_{\mathbb R^3}\theta^2|\nabla v_t|^2
\,dx+
2t^{-2}\int_{\operatorname{supp}\nabla\theta}v_t^2\,dx\\
&\le
2\int_{\mathbb R^3}|\nabla v_t|^2
\,dx+
2t^{-2}\sum_{i=1}^2\int_{S_i}v_t^2 \,dx.
\end{aligned}
\]
Moreover,
\[
\int_{\mathbb R^3}|\nabla v_t|^2
\,dx=
4\pi\int_t^\infty s^{-2}\,ds
=
\frac{4\pi}{t} .
\]
From~\eqref{eq:vt-moment} with \(p=2\), we obtain
\[
\begin{aligned}
\int_{S_i}v_t^2
\,dx&=
2\int_t^\infty\frac{F(s)}{s^3}\,ds\\
&\le
Ct\int_t^r\frac{ds}{s}
+
Cr^2t\int_r^\infty\frac{ds}{s^3}\\
&\le
Ct\Bigl(1+\log\frac rt\Bigr),
\end{aligned}
\]
where the middle inequality uses \(F(s)\le Cts^2\) for \(s\le r\) and
\(F(s)\le|S_i|\le Cr^2t\) for \(s\ge r\). Therefore
\begin{equation}\label{eq:trunc-gradient}
\|\nabla\psi\|_2
\leq
Ct^{-1/2}
\bigl[1+\log(r/t)\bigr]^{1/2}.
\end{equation}

From H\"older's inequality and~\eqref{eq:trunc-error}, we have
\[
\left|
\int_{\mathbb R^3}\zeta(k_x-\psi)
\,dx\right|
\leq
C\|\zeta\|_{5/3}t^{1/5}.
\]
Then combining~\eqref{eq:coulomb-duality}
and~\eqref{eq:trunc-gradient}, we obtain
\[
\left|
\int_{\mathbb R^3}\zeta\psi
\,dx\right|
\leq
C\mathcal{D}(\zeta,\zeta)^{1/2}
t^{-1/2}
\bigl[1+\log(r/t)\bigr]^{1/2}.
\]
Therefore,
the two bounds above give
\[
\begin{aligned}
\Big|\int_{K_r}\frac{\zeta(y)}{|x-y|}\,dy\Big|
&=
\int_{\mathbb R^3}\zeta(k_x-\psi)
\,dx+
\int_{\mathbb R^3}\zeta\psi\,dx\\
&\le
\Big|\int_{\mathbb R^3}\zeta(k_x-\psi)\,dx\Big|
+
\Big|\int_{\mathbb R^3}\zeta\psi\,dx\Big|\\
&\le
C\|\zeta\|_{5/3}t^{1/5}
+
C\mathcal{D}(\zeta,\zeta)^{1/2}t^{-1/2}\bigl[1+\log(r/t)\bigr]^{1/2},
\end{aligned}
\]
which is the bound of the theorem.
\end{proof}

The following comparison principle propagates a bound on the sphere
\(\{\delta=r\}\) to the exterior region.
\begin{lemma}\label{lem:exterior-barrier}
Let \(\zeta\in L^{5/3}(\mathbb R^3)\). For \(r>0\) and \(B\geq0\), if
\[
\Big|\int_{K_r}\frac{\zeta(y)}{|x-y|}\,dy\Big|\le B
\qquad\text{on }\{\delta=r\} ,
\]
then
\[
\delta(x)\Bigl|\int_{K_r}\frac{\zeta(y)}{|x-y|}\,dy\Bigr|\le2rB
\qquad\text{for }\delta(x)\geq r,
\qquad\text{and}\qquad
\Bigl|\int_{K_r}\zeta\,dx\Bigr|\le2rB .
\]
\end{lemma}

\begin{proof}
Write \(w=\mathbf1_{K_r}\zeta\). Since \(K_r\) is bounded,
\(w\in L^1\cap L^{5/3}\). The capped potentials
\[
\int_{K_r}\frac{\zeta(y)}{\max\{|x-y|,\varepsilon\}}\,dy
\]
are continuous, and
\[
\sup_{x\in\mathbb R^3}\Bigl|
\int_{K_r}\frac{\zeta(y)}{|x-y|}\,dy
-
\int_{K_r}\frac{\zeta(y)}{\max\{|x-y|,\varepsilon\}}\,dy
\Bigr|
\le C\|w\|_{5/3}\varepsilon^{1/5} .
\]
Since the bound tends to \(0\) as \(\varepsilon\downarrow0\), the potential
\(\int_{K_r}\frac{\zeta(y)}{|x-y|}\,dy\) is a uniform limit of continuous
functions of \(x\), hence continuous.
Since \(K_r\)
stays at positive distance from every compact subset of \(K_r^c\),
differentiation under the integral sign is justified there. Hence
\[
\Delta\left[\int_{K_r}\frac{\zeta(y)}{|x-y|}\,dy\right]=0
\qquad\text{in }K_r^c .
\]
Moreover, since \(w\) has bounded support, dominated convergence gives
\[
\int_{K_r}\frac{\zeta(y)}{|x-y|}\,dy\longrightarrow0,
\qquad
|x|\int_{K_r}\frac{\zeta(y)}{|x-y|}\,dy\longrightarrow\int_{\mathbb R^3} w
\,dx\quad(|x|\to\infty).
\]
Set
\[
b_r(x)=\frac{r}{|x-R_1|}+\frac{r}{|x-R_2|}.
\]
The function \(b_r\) is harmonic in \(\{\delta>r\}\), tends to zero at infinity,
and satisfies \(b_r\geq1\) on \(\{\delta=r\}\). Moreover
\[
\delta(x)b_r(x)\leq2r.
\]

By hypothesis,
\[
\Bigl|\int_{K_r}\frac{\zeta(y)}{|x-y|}\,dy\Bigr|\le B
\qquad\text{on }\{\delta=r\} .
\]
Then we apply the exterior maximum principle to
\[
\int_{K_r}\frac{\zeta(y)}{|x-y|}\,dy-Bb_r
\quad\text{and}\quad
-\int_{K_r}\frac{\zeta(y)}{|x-y|}\,dy-Bb_r
\]
and obtain
\begin{equation}\label{eq:barrier-comparison}
\Bigl|\int_{K_r}\frac{\zeta(y)}{|x-y|}\,dy\Bigr|\leq Bb_r(x)
\qquad\text{for }\delta(x)>r .
\end{equation}
The comparison in the unbounded region follows by
the maximum principle on \(K_r^c\cap B_L(0)\), then we let
\(L\to\infty\). On \(\{\delta=r\}\) it follows directly from the hypothesis
and \(b_r\geq1\).
Multiplying by \(\delta(x)\) yields
\[
\delta(x)\Bigl|\int_{K_r}\frac{\zeta(y)}{|x-y|}\,dy\Bigr|\leq2rB
\qquad\text{for }\delta(x)\geq r .
\]

Since
\[
|x|\int_{K_r}\frac{\zeta(y)}{|x-y|}\,dy\longrightarrow\int_{K_r}\zeta
\,dx\qquad\text{and}\qquad
|x|b_r(x)\longrightarrow2r
\qquad\text{as }|x|\to\infty,
\]
letting \(|x|\to\infty\) in~\eqref{eq:barrier-comparison} proves the last
assertion.
\end{proof}

\section{Exterior \texorpdfstring{\(L^1\)}{L1}-Estimates}
\label{sec:exterior-estimates}

In this section we prove an exterior \(L^1\)-estimate for the number of
electrons outside \(K_r\).

\begin{proposition}\label{prop:screened-quadratic}
For every \(r>0\), \(0<\lambda\leq1/2\), and \(h>0\),
\[
\frac{m_{(1+\lambda)r}^2}{16}
\leq
\frac14 S_rm_r+2h\mathcal D_r
+\left(\alpha h+\frac{\beta}{\lambda^2h}\right)m_r,
\]
where
\[
\alpha=\frac{c_{\mathrm D}^2}{4c_{\mathrm{TF}}},
\qquad
\beta=2c_{\mathrm W}\pi^2 .
\]
\end{proposition}

\begin{proof}
We use the folded map \(F\) from Lemma~\ref{lem:folded-localization-rough}
and set
\[
g(x)=\min\left\{1,\max\left\{0,\frac{\delta(x)-r}{\lambda r}\right\}\right\},
\qquad \Psi(x)=g(x)F(x).
\]
Then \(|\Psi(x)|\leq\delta(x)\), \(\Psi=0\) in \(K_r\), and
\(\Psi=F\) in \(\{\delta\geq(1+\lambda)r\}\). The radial map
\(z\mapsto\min\{1,\max\{0,(|z|-r)/(\lambda r)\}\}z\)
is Lipschitz with constant at most
\(1+(1+\lambda)/\lambda=2+1/\lambda\leq2/\lambda\). Composition with the
\(1\)-Lipschitz map \(F\) gives the same bound for \(\Psi\).

We use the same cutoff functions as in the proof of
Lemma~\ref{lem:folded-localization-rough} and define
\[
q(s)=\min\{1,\max\{0,s\}\},\quad
a_0(s)=\cos(\pi q(s)/2),\quad
a_1(s)=\sin(\pi q(s)/2).
\]
For \(\nu\in\mathbb S^2\) and \(l>0\), set
\[
p_\nu(x)=\nu\cdot\Psi(x),\qquad
f_j(x)=a_j\bigl((p_\nu(x)-l)/h\bigr),\qquad
A_{\nu,l}=\{l<p_\nu<l+h\}.
\]
These functions form a Lipschitz partition,
\[
f_0^2+f_1^2=1,
\qquad
f_0=1,\quad f_1=0
\qquad\text{in }K_r .
\]
The gradients are supported in \(A_{\nu,l}\subset K_r^c\) and satisfy
\[
|\nabla f_0|^2+|\nabla f_1|^2
\leq \frac{2\pi^2}{\lambda^2h^2}\mathbf1_{A_{\nu,l}}
\qquad\text{a.e.}
\]
Since \(f_0=1\) and \(f_1=0\) in \(K_r\),~\eqref{eq:partition-localization}
gives
\[
\begin{aligned}
&\iint_{\mathbb R^3\times\mathbb R^3}\frac{f_0(x)^2\rho_0(x)f_1(y)^2\rho_0(y)}{|x-y|}\,dx\,dy\\
&=
\iint_{\mathbb R^3\times\mathbb R^3}\frac{e_r(x)e_r(y)}{|x-y|}f_0(x)^2f_1(y)^2\,dx\,dy
+\int_{\mathbb R^3} U_rf_1^2e_r \,dx,
\end{aligned}
\]
Consequently, since \(\int_{\mathbb R^3} Vf_1^2e_r\,dx-\int_{\mathbb R^3} U_rf_1^2e_r\,dx=\int_{\mathbb R^3}\Phi_rf_1^2e_r\,dx\),
we obtain
\begin{equation}\label{eq:screened-cutoff-local}
\begin{aligned}
I_{\nu,l}
&:=\iint_{\mathbb R^3\times\mathbb R^3}\frac{e_r(x)e_r(y)}{|x-y|}f_0(x)^2f_1(y)^2\,dx\,dy\\
&\leq\int_{\mathbb R^3}[\Phi_r(x)]_+f_1^2e_r
\,dx+\left(\alpha+\frac{\beta}{\lambda^2h^2}\right)
\int_{A_{\nu,l}}e_r\,dx.
\end{aligned}
\end{equation}
All terms in~\eqref{eq:screened-cutoff-local} are finite, since
\(\mathcal D_r<\infty\) and
\[
0\leq\delta(x)[\Phi_r(x)]_+\leq\delta(x)V(x)\leq Z
\qquad\text{in }K_r^c .
\]

For \(a,b\in\mathbb R\), define
\[
L_h(a,b)=\int_0^\infty
a_0((a-l)/h)^2a_1((b-l)/h)^2\,dl ,
\]
the length of the set of levels \(l\) at which both \(a_0((a-l)/h)\) and
\(a_1((b-l)/h)\) equal \(1\).
From Tonelli's theorem, we can integrate~\eqref{eq:screened-cutoff-local}
over \(l>0\) to obtain
\begin{equation}\label{eq:screened-height}
J_\nu:=\iint_{\mathbb R^3\times\mathbb R^3}\frac{e_r(x)e_r(y)}{|x-y|}
L_h(p_\nu(x),p_\nu(y))\,dx\,dy
\leq Q_\nu+\left(\alpha h+\frac{\beta}{\lambda^2h}\right)m_r,
\end{equation}
where
\[
Q_\nu=\int_{\mathbb R^3}(p_\nu)_+[\Phi_r(x)]_+e_r \,dx.
\]
Here we used
\[
\int_0^\infty f_1(x)^2\,dl\leq(p_\nu(x))_+
\quad\text{and}\quad
\int_0^\infty\mathbf1_{A_{\nu,l}}(x)\,dl\leq h .
\]

We next symmetrize before discarding any Coulomb factor. Recall the function
\(H(a,b)=(b-\max\{a,0\})_+\) from the proof of
Lemma~\ref{lem:folded-localization-rough}. Since
\[
a_0((a-l)/h)^2a_1((b-l)/h)^2=1
\qquad\text{for }\max\{a,0\}<l<b-h ,
\]
we have \(L_h(a,b)\geq H(a,b)-h\).
Moreover,
\[
\begin{aligned}
|a-b|&=H(a,b)+H(b,a)+H(-a,-b)+H(-b,-a)\\
&\leq L_h(a,b)+L_h(b,a)+L_h(-a,-b)+L_h(-b,-a)+4h .
\end{aligned}
\]
We integrate the resulting inequality against \(e_r(x)e_r(y)/|x-y|\,dx\,dy\)
and use the symmetry in \(x,y\) to obtain
\[
\frac14\iint_{\mathbb R^3\times\mathbb R^3}\frac{e_r(x)e_r(y)}{|x-y|}
|p_\nu(x)\,dx\,dy-p_\nu(y)|\,dx\,dy
\leq\frac{J_\nu+J_{-\nu}}2+2h\mathcal{D}_r,
\]
where we used \(\iint_{\mathbb R^3\times\mathbb R^3} e_r(x)e_r(y)/|x\,dx\,dy-y|\,dx\,dy=2\mathcal D_r\).

Then we average this inequality and~\eqref{eq:screened-height} with respect to
the probability measure \(d\omega=dS/(4\pi)\) on \(\mathbb S^2\) to obtain
\begin{equation}\label{eq:screened-averaged}
\begin{aligned}
&\frac18\sum_{j=1}^2\iint_{\Omega_j\cap K_{(1+\lambda)r}^c}
\frac{e_r(x)e_r(y)}{|x-y|}|\Psi(x)\,dx\,dy-\Psi(y)|\,dx\,dy\\
&\qquad\leq\frac18\iint_{\mathbb R^3\times\mathbb R^3}\frac{e_r(x)e_r(y)}{|x-y|}|\Psi(x)\,dx\,dy-\Psi(y)|\,dx\,dy\\
&\qquad\leq\frac14 S_rm_r+2h\mathcal D_r
+\left(\alpha h+\frac{\beta}{\lambda^2h}\right)m_r .
\end{aligned}
\end{equation}
Here we use
\[
\int_{\mathbb S^2}|\nu\cdot z|\,d\omega=\frac{|z|}{2},
\qquad
\int_{\mathbb S^2}(\nu\cdot z)_+\,d\omega=\frac{|z|}{4} ,
\]
together with \(|\Psi|[\Phi_r(x)]_+\leq S_r\).

For almost every \(x,y\in\Omega_j\cap K_{(1+\lambda)r}^c\) we have
\[
e_r=\rho_0 ,
\qquad
\Psi=F ,
\qquad
|F(x)-F(y)|=|x-y| .
\]
Therefore
\[
\sum_{j=1}^2\iint_{\Omega_j\cap K_{(1+\lambda)r}^c}
\frac{e_r(x)e_r(y)}{|x-y|}|\Psi(x)\,dx\,dy-\Psi(y)|\,dx\,dy
=\tilde M_1^2+\tilde M_2^2 ,
\]
where \(\tilde M_j=\int_{\Omega_j\cap K_{(1+\lambda)r}^c}\rho_0\,dx\). Since
\[
\tilde M_1^2+\tilde M_2^2\ge\frac12(\tilde M_1+\tilde M_2)^2
=\frac12m_{(1+\lambda)r}^2 .
\]
the estimate~\eqref{eq:screened-averaged} gives
\[
\begin{aligned}
\frac1{16}m_{(1+\lambda)r}^2
&\le\frac18\bigl(\tilde M_1^2+\tilde M_2^2\bigr)\\
&\le\frac14 S_rm_r+2h\mathcal D_r
+\left(\alpha h+\frac{\beta}{\lambda^2h}\right)m_r .
\end{aligned}
\]
\end{proof}

\begin{corollary}\label{cor:exterior-mass}
There exists \(C>0\) such that, for every \(r>0\), \(0<\lambda\leq1/2\), and
\(h>0\),
\[
m_{(1+\lambda)r}\leq C\left[
\int_{K_{(1+\lambda)r}\setminus K_r}\rho_0\,dx+S_r+h+\lambda^{-2}h^{-1}
+\sqrt{h\mathcal D_r}\right].
\]
\end{corollary}

\begin{proof}
Without loss of generality we may assume that
\[
\int_{K_{(1+\lambda)r}\setminus K_r}\rho_0\,dx<\frac12m_{(1+\lambda)r} .
\]
Consequently, we gain
\[
m_r=m_{(1+\lambda)r}+\int_{K_{(1+\lambda)r}\setminus K_r}\rho_0
\,dx<\frac32m_{(1+\lambda)r} .
\]

Then from Proposition~\ref{prop:screened-quadratic}, we obtain
\[
\begin{aligned}
\frac{1}{16}m_{(1+\lambda)r}^2
&\le\frac14 S_rm_r+2h\mathcal D_r
+\left(\alpha h+\frac{\beta}{\lambda^2h}\right)m_r\\
&\le C\bigl(S_r+h+\lambda^{-2}h^{-1}\bigr)m_{(1+\lambda)r}+Ch\mathcal D_r .
\end{aligned}
\]
Solving the quadratic inequality for \(m_{(1+\lambda)r}\), we obtain
\[
m_{(1+\lambda)r}\le C\left[
S_r+h+\lambda^{-2}h^{-1}+\sqrt{h\mathcal D_r}\right].
\]

This gives the asserted bound
\[
m_{(1+\lambda)r}\leq C\left[
\int_{K_{(1+\lambda)r}\setminus K_r}\rho_0\,dx+S_r+h+\lambda^{-2}h^{-1}
+\sqrt{h\mathcal D_r}\right].
\]
\end{proof}

\section{Proof of the Main Theorem}
\label{sec:main-proof}

In this section we prove Theorem~\ref{thm:main}.

Without loss of generality we assume \(N\geq Z\).
For the remainder of the proof, let
\[
\varrho=\rho_0-\rho^{\mathrm{TF}}.
\]
For \(Z\geq2\), Proposition~\ref{prop:tf-comparison} gives
\[
\mathcal{D}(\varrho,\varrho)\leq CZ^{25/11}, \qquad \|\varrho\|_{5/3}\leq
CZ^{7/5}.
\]

Let
\[
1\leq R\leq Z^{1/3}, \qquad r=RZ^{-1/3}, \qquad t=Z^{-29/77}.
\]
Then
\(0<r\leq1\) and \(t = R^{-1}Z^{-10/231} r.\)

For simplicity, we write \(q=kr\), where \(k=1\) or \(3/2\), and every
statement in \(q\) below holds for both radii. Since \(0<t<r/4\leq q/4\) for
all \(Z\) large enough, we apply Theorem~\ref{thm:screened-coulomb} to
\(\varrho\) and obtain
\begin{equation}\label{eq:BZ}
\begin{aligned}
\Bigl|\int_{K_q}\frac{\varrho(y)}{|x-y|}\,dy\Bigr|
&\leq C \left[ Z^{25/22}t^{-1/2} \bigl[1+\log(q/t)\bigr]^{1/2}
+ Z^{7/5}t^{1/5} \right] \\
&\leq CZ^{102/77}\bigl(1+\log R+\log Z\bigr)^{1/2},
\qquad x\in\mathbb R^3 .
\end{aligned}
\end{equation}
The second inequality uses \(q/t=kRZ^{10/231}\) and
\[
1+\log(q/t) = 1+\log k+\log R+\frac{10}{231}\log Z \leq 2(1+\log R+\log Z).
\]

Define
\[
Q_q = \int_{K_q}(\rho_0-\rho^{\mathrm{TF}}) \,dx=\int_{K_q}\varrho\,dx.
\]
Fix one constant \(C\) in~\eqref{eq:BZ} and set
\[
B_Z=CZ^{102/77} \bigl(1+\log R+\log Z\bigr)^{1/2}.
\]
Then Lemma~\ref{lem:exterior-barrier} gives
\[
|Q_q|\leq 2qB_Z, \qquad q=r,\frac32 r.
\]
Using \(\int_{\mathbb R^3}\rho^{\mathrm{TF}}\,dx=Z\), we obtain
\begin{equation}\label{eq:core-mass}
\int_{K_r}\rho_0 \,dx\leq Z+2rB_Z.
\end{equation}

From Corollary~\ref{cor:tf-tail}, we have
\begin{equation}\label{eq:shell-mass}
\int_{K_{3r/2}\setminus K_r}\rho_0
\,dx=\int_{K_{3r/2}\setminus K_r}\rho^{\mathrm{TF}}\,dx+Q_{3r/2}-Q_r
\leq Cr^{-3}+5rB_Z .
\end{equation}

Since
\[
\Phi_r=\Phi_r^{\mathrm{TF}}
-\int_{K_r}\frac{\varrho(y)}{|x-y|}\,dy ,
\]
Corollary~\ref{cor:tf-tail}, Lemma~\ref{lem:exterior-barrier},
and~\eqref{eq:BZ} give
\begin{equation}\label{eq:Sr}
S_r=\sup_{\delta(x)\geq r}\delta(x)[\Phi_r(x)]_+\leq Cr^{-3}+2rB_Z .
\end{equation}

We combine Proposition~\ref{prop:hartree-apriori} with
the estimate~\eqref{eq:rough-mass-detailed} and obtain
\[
\mathcal{D}_r\leq\mathcal{D}(\rho_0,\rho_0)\leq CZ^{7/3}.
\]
We apply Corollary~\ref{cor:exterior-mass} with \(\lambda=\frac12\) and
\(h=r^3\),
\[
m_{3r/2}\leq C\left[\int_{K_{3r/2}\setminus K_r}\rho_0\,dx+S_r+r^3+4r^{-3}
+\sqrt{r^3\mathcal D_r}\right].
\]
Since \(0<r\leq1\) and \(\mathcal{D}_r\leq CZ^{7/3}\), we have
\[
r^3+4r^{-3}\leq5r^{-3}, \qquad \sqrt{r^3\mathcal D_r}\leq Cr^{3/2}Z^{7/6}.
\]
Together with~\eqref{eq:shell-mass} and~\eqref{eq:Sr}, we obtain
\begin{equation}\label{eq:m32r}
m_{3r/2} \leq C \left( r^{-3} + rB_Z + r^{3/2}Z^{7/6} \right).
\end{equation}

Since \(m_r=m_{3r/2}+\int_{K_{3r/2}\setminus K_r}\rho_0\,dx\),
equations~\eqref{eq:shell-mass} and~\eqref{eq:m32r} give
\[
m_r\leq C \left( r^{-3} + rB_Z + r^{3/2}Z^{7/6} \right).
\]
Since \(N=\int_{K_r}\rho_0\,dx+m_r\), equation~\eqref{eq:core-mass} then yields
\begin{equation}\label{eq:mr}
N \leq Z + C \left( r^{-3} + rB_Z + r^{3/2}Z^{7/6} \right).
\end{equation}
Dividing~\eqref{eq:mr} by \(Z\) and recalling \(r=RZ^{-1/3}\), we obtain
\begin{equation}\label{eq:pre-optimization}
\frac NZ \leq 1 + CR^{-3} + CRZ^{-2/231} \bigl(1+\log R+\log Z\bigr)^{1/2} +
CR^{3/2}Z^{-1/3}.
\end{equation}

We choose \(R=Z^{1/462}(1+\log Z)^{-1/8}\) to balance the first two terms
in~\eqref{eq:pre-optimization}. Then we gain
\[
\begin{aligned}
\frac NZ
&\leq1+CZ^{-1/154}(1+\log Z)^{3/8}\\
&\qquad+CZ^{-1/154}(1+\log Z)^{-1/8}\bigl(1+\log R+\log Z\bigr)^{1/2}\\
&\qquad+CZ^{-305/924}(1+\log Z)^{-3/16}\\
&\leq1+CZ^{-1/154}(1+\log Z)^{3/8}.
\end{aligned}
\]
Equivalently, \(N\leq Z+CZ^{153/154}(1+\log Z)^{3/8}\).
This completes the proof of Theorem~\ref{thm:main}.

\bibliographystyle{amsplain}
\bibliography{references}

\end{document}